%% file: main.tex
\documentclass[11pt,letterpaper]{article}

\usepackage[utf8]{inputenc}
\usepackage[english]{babel}

\usepackage[dvipsnames,usenames]{xcolor}
\usepackage[colorlinks=true,pdfpagemode=UseNone,urlcolor=RoyalBlue,linkcolor=RoyalBlue,citecolor=OliveGreen,pdfstartview=FitH]{hyperref}

\usepackage{tikz}
\usetikzlibrary{shapes.geometric}

\usepackage{subcaption}

\usepackage[margin=1in]{geometry}
\usepackage{amsfonts,amsmath,amsthm,thmtools,amssymb}
\usepackage{mathtools}

\usepackage{nicefrac}

\declaretheorem{theorem}
\declaretheorem[sibling=theorem]{lemma}
\declaretheorem[sibling=theorem]{corollary}
\declaretheorem{remark}

\usepackage{todonotes}
\usepackage{booktabs}
\usepackage{caption}

\usepackage[numbers, sort]{natbib}

\usepackage{tcolorbox}

\usepackage{enumitem}
\usepackage{xspace}
\usepackage{listings}
\usepackage{xcolor}

\input{macro}

\title{
	Online Matching Meets Sampling Without Replacement
}

\author{
	Zhiyi Huang\thanks{The University of Hong Kong. \texttt{zhiyi@cs.hku.hk}} \and
	Chui Shan Lee\thanks{The University of Hong Kong. \texttt{cslee@cs.hku.hk}} \and
	Jianqiao Lu\thanks{The University of Hong Kong. \texttt{jqlu@cs.hku.hk}} \and
	Xinkai Shu\thanks{The University of Hong Kong. \texttt{xkshu@cs.hku.hk}}
}

\date{October 2024}

\begin{document}
	
	\begin{titlepage}
		\thispagestyle{empty}
		\maketitle
		\begin{abstract}
			\thispagestyle{empty}
			\input{abstract}
		\end{abstract}
	\end{titlepage}

	\input{introduction}
	\input{preliminaries}
	\input{SWOR}

	\input{balanceSWOR}

	\input{discussion}

	\bibliographystyle{plainnat}

	\appendix
	\input{app-experiment}
	\input{app-swor}
	\input{app-balance-swor}
	\input{app-regularized-greedy-weighted}
\end{document}

%% file: macro.tex
\newcommand{\E}{\mathbf{E}}
\renewcommand{\Pr}{\mathbf{Pr}}

\newcommand{\alg}{A}
\newcommand{\opt}{\textsc{Opt}}
\newcommand{\pot}{\Phi}

\newcommand{\Sref}[1]{\hyperref[#1]{\S\ref*{#1}}}

\newcommand{\R}{\mathbb{R}}

\newcommand{\suggestedmatch}{\textsc{Suggested Matching}\xspace}
\newcommand{\multisuggestedmatch}{\textsc{Multistage Suggested Matching}\xspace}

\newcommand{\balance}{\textsc{Balance}\xspace}
\newcommand{\ranking}{\textsc{Ranking}\xspace}
\newcommand{\feldman}{\textsc{Feldman et al}\xspace}
\newcommand{\bahmani}{\textsc{Bahmani Kapralov}\xspace}
\newcommand{\manshadi}{\textsc{Manshadi et al}\xspace}
\newcommand{\jailletlu}{\textsc{Jaillet Lu (Integral)}\xspace}
\newcommand{\jlnonint}{\textsc{Jaillet Lu}\xspace}
\newcommand{\brubach}{\textsc{Brubach et al}\xspace}
\newcommand{\mindegree}{\textsc{Min Degree}\xspace}
\newcommand{\balanceocs}{\textsc{Balance OCS}\xspace}
\newcommand{\poissonocs}{\textsc{Poisson OCS}\xspace}
\newcommand{\tophalf}{\textsc{Top Half Sampling}\xspace}

\newcommand{\swor}{\textsc{Stochastic SWOR}\xspace}
\newcommand{\reggreedy}{\textsc{Regularized Greedy}\xspace}
\newcommand{\balanceswor}{\textsc{Balance SWOR}\xspace}

\newcommand{\correlated}{\textsc{Correlated Sampling}\xspace}
\newcommand{\haeupler}{\textsc{Haeupler et al}\xspace}

%% file: abstract.tex
Sampling without replacement is a natural online rounding strategy for converting fractional bipartite matching into an integral one.
In Online Bipartite Matching, we can use the Balance algorithm to fractionally match each online vertex, and then sample an unmatched offline neighbor with probability proportional to the fractional matching.
In Online Stochastic Matching, we can take the solution to a linear program relaxation as a reference, and then match each online vertex to an unmatched offline neighbor with probability proportional to the fractional matching of the online vertex's type.
On the one hand, we find empirical evidence that online matching algorithms based on sampling without replacement outperform existing algorithms.
On the other hand, the literature offers little theoretical understanding of the power of sampling without replacement in online matching problems.

This paper fills the gap in the literature by giving the first non-trivial competitive analyses of sampling without replacement for online matching problems.
In Online Stochastic Matching, we develop a potential function analysis framework to show that sampling without replacement is at least $0.707$-competitive.
The new analysis framework further allows us to derandomize the algorithm to obtain the first polynomial-time deterministic algorithm that breaks the $1-\frac{1}{e}$ barrier.
In Online Bipartite Matching, we show that sampling without replacement provides provable online correlated selection guarantees when the selection probabilities correspond to the fractional matching chosen by the Balance algorithm.
As a result, we prove that sampling without replacement is at least $0.513$-competitive for Online Bipartite Matching.

%% file: introduction.tex
\section{Introduction}
\label{sec:introduction}

Matching problems have been a central topic in theoretical computer science for decades.
They arise in a wide range of applications.
Online advertising platforms match users and advertisers.
Ride-hailing services assign drivers to riders.
Kidney exchange programs pair patients and donors.
In many of these scenarios, the goal is to match a set of agents or objects to another set of agents or objects, subject to some constraints or preferences.
The success of these applications heavily relies on the efficiency and effectiveness of the algorithms used to solve matching problems.

With the rise of online platforms and their need for real-time decision-making, online matching algorithms have gained much attention in recent years.
A prominent model is the Online Bipartite Matching problem proposed by \citet{Karp:STOC:1990}.
It considers a bipartite graph with online vertices on the left and offline vertices on the right.
The online vertices arrive one at a time, and the online algorithms need to decide how to match each online vertex as it arrives without information about future online vertices and their edges.

Researchers have developed two approaches to design algorithms for Online Bipartite Matching.
On the one hand, \citet{Karp:STOC:1990} proposed the \ranking algorithm and proved that it achieves the optimal $1-\frac{1}{e}$ competitive ratio.
It randomly ranks the offline vertices at the beginning, and then greedily matches each online vertex to the highest-ranked unmatched neighbor.

On the other hand, \citet{Kalyanasundaram:TCS:2000} considered the \balance algorithm (a.k.a., \textsc{Water-Filling} or \textsc{Water-Level}).
It may be viewed as a fractional online matching algorithm that splits one unit of each online vertex among its neighbors while balancing the total amounts of online vertices allocated to different offline vertices.
Moreover, there have been recent studies on online rounding algorithms that convert \balance and its variants into integral algorithms for Online Bipartite Matching.
For instance, \citet{Gao:FOCS:2021} designed the \balanceocs algorithm building on the online correlated selection technique, which was first proposed by \citet*{Fahrbach:2022:JACM}, then further developed in a series of works \cite{Huang:FOCS:2020, Shin:ARXIV:2021, Gao:FOCS:2021, Blanc:FOCS:2022}.
\citet*{BuchbinderNW:SODA:2023} and \citet*{Naor:SW:arXiv:2023} also studied online rounding algorithms for matching and how to use negative correlation to improve  rounding efficiency.

Despite these efforts, there has been no theoretical understanding of arguably the most natural rounding algorithm, which samples an unmatched neighbor for each online vertex with probability proportional to the fractional allocation of \balance.
Since this approach effectively picks an offline vertex via sampling without replacement (SWOR) with sampling masses given by \balance, we call the resulting online matching algorithm \balanceswor.

We next consider a related problem known as Online Stochastic Matching, which was first introduced by \citet*{FeldmanMMM:FOCS:2009}.
It is also about online matching on bipartite graphs but the online vertices therein are generated by a random process that is known to the algorithms from the beginning.

There is a rich literature on Online Stochastic Matching, in which all algorithms follow a two-step framework.
First, algorithms find a fractional bipartite matching between the offline vertices and the online vertex types, which represents the likelihood that they are matched in the optimal matching of the randomly generated bipartite graph.
This fractional matching is often computed by solving a linear program (LP) constructed based on the stochastic process that generates online vertices.
Then, when an online vertex arrives, algorithms use different online rounding strategies to make a matching decision, based on the current set of unmatched offline vertices and the fractional matching of the online vertex's type.

For example, the \suggestedmatch algorithm~\cite{FeldmanMMM:FOCS:2009} first samples an offline neighbor (regardless of whether it is matched) with probabilities equal to the fractional matching of the online vertex's type.
It then matches the online vertex to this neighbor if it is still unmatched.
This simple approach is shown to be $(1-\frac{1}{e})$-competitive for Online Stochastic Matching.

Since then, many efforts \cite{FeldmanMMM:FOCS:2009, Bahmani:ESA:2010, Haeupler:WINE:2011, ManshadiOS:MOR:2012, Jaillet:MOR:2014, Brubach:ESA:2016, HuangS:STOC:2021, Yan:arXiv:2022} have been devoted to improving the online rounding strategy by sampling a second offline neighbor, usually correlated to the sampling result of the first neighbor, so that algorithms could try the second neighbor when the first is already matched.
This is known as the power of two choices in Online Stochastic Matching, resulting in online algorithms with competitive ratios strictly better than $1-\frac{1}{e}$.
\citet*{HuangSY:STOC:2022} and \citet*{TangWW:STOC:2022} explored the power of multiple choices in Online Stochastic Matching by adopting online correlated selection algorithms as their rounding strategies.
The \poissonocs algorithm by \citet*{HuangSY:STOC:2022} achieves the state-of-the-art competitive ratio of $0.716$.

Once again, despite the long line of research, we still have little understanding on arguably the most natural rounding algorithm, which samples an unmatched neighbor with probability proportional to the fractional matching of the online vertex's type.
We will refer to this algorithm as \swor, where SWOR stands for sampling without replacement.

\subsection{Our Contributions}

We first examine the mentioned \balanceswor and \swor algorithms through empirical experiments to determine whether it is worthwhile to further study them theoretically.
Our experiments compare them with existing online matching algorithms on six real-world datasets that were used in a recent empirical study by \citet*{Borodin:JEA:2020}.

To our surprise, \balanceswor and \swor outperform existing algorithms that admit non-trivial theoretical guarantees.
Among algorithms for Online Bipartite Matching, \balanceswor achieves slightly better results compared to both \balanceocs and \ranking.
Among algorithms for Online Stochastic Matching, \swor and a derandomized version called \reggreedy, which we will explain shortly, perform slightly better than \poissonocs, which in turn shows substantially better results than other algorithms from the Online Stochastic Matching literature.
See Table~\ref{tab:intro-experiment} for some of our empirical results, and Appendix~\ref{sec:experiment} for more details.

\begin{table}[htbp]
	\caption{Empirical results for the algorithms considered in this paper (in bold) and top-performing existing algorithms with known theoretical guarantees. See Appendix~\ref{sec:experiment} for the full results.}
	\label{tab:intro-experiment}
	\begin{center}
		\begin{tabular}{lcccccc}
			\toprule
			Algorithm &Caltech36 &Reed98 &CE-GN &CE-PG &beause &mbeaflw\\
			\midrule
			\textbf{Regularized Greedy} & 0.928 & 0.929 & 0.984 & 0.990 & 0.962 & 0.966\\
			\textbf{Stochastic SWOR} & 0.929 & 0.927 & 0.958 & 0.962 & 0.959 & 0.975\\
			Poisson OCS & 0.929 & 0.926 & 0.957 & 0.960 & 0.958 & 0.974\\[1ex]
			\textbf{Balance SWOR} & 0.874 & 0.873 & 0.943 & 0.950 & 0.943 & 0.971\\
			Balance OCS & 0.871 & 0.870 & 0.942 & 0.949 & 0.942 & 0.970\\
			Ranking & 0.859 & 0.859 & 0.934 & 0.944 & 0.936 & 0.966\\
			\bottomrule
		\end{tabular}
	\end{center}
\end{table}

\vspace{-1.5em}
Conceptually, the empirical experiments suggest that a simple rounding strategy like sampling without replacement may be sufficient in the context of online matching problems, and may even be slightly superior compared to more involved techniques like online correlated selection.
The rest of the paper is therefore devoted to providing non-trivial competitive analyses of \balanceswor and \swor.

\paragraph{Competitive Analysis of Stochastic SWOR.}
This paper gives the first non-trivial analysis of \swor, showing that it is at least $0.707$-competitive for Online Stochastic Matching.
Although it is worse than the state-of-the-art $0.716$ ratio \cite{HuangSY:STOC:2022}, it surpasses a natural barrier $0.706$ suffered by all algorithms prior to the recent introduction of the Natural LP by \cite{HuangS:STOC:2021, HuangSY:STOC:2022}.

Our analysis develops a \emph{potential function analysis framework} that is novel in the context of Online Stochastic Matching.
Up to normalization, we may assume without loss of generality that online vertices arrive within the time horizon $0$ to $1$.
We design a potential function $\pot(t)$ for any time $0 \le t \le 1$, which depends on the LP-based fractional matching and the current set of unmatched offline vertices.
Intuitively, $\pot(t)$ is a lower bound for the expected number of matches the online algorithm would get after time $t$.
To substantiate this intuition, $\pot(t)$ shall satisfy three conditions.
The first one is a boundary condition at the beginning such that $\pot(0) \ge \Gamma \cdot \opt$, where $\Gamma$ is the competitive ratio and $\opt$ is the optimal objective.
The second one is a boundary condition at the end, namely, $\pot(1) = 0$.
Finally, we will show that the online algorithm's expected match rate weakly dominates the potential function's expected decrease rate, i.e., $\E\,\big[ \alg(t) + \pot(t) \big]$ is non-decreasing in $0 \le t \le 1$ where $\alg(t)$ denotes the size of algorithm's matching up to time $t$.

As natural as it may seem, this potential function analysis framework has never been employed in the literature of Online Stochastic Matching.
We argue that it might be the right approach.
In fact, the folklore yet computationally intractable optimal algorithm for Online Stochastic Matching fits into this framework.
Here we let $\pot(t)$ be the best expected matching size achievable by online algorithms after time $t$, which can be computed in exponential time by backward induction.
When an online vertex arrives, the optimal online algorithm matches it to an unmatched offline neighbor such that the potential function decreases the least.

The challenge is therefore to design a computationally tractable potential function for the stated competitive ratio.
Our potential computes the contribution of each online vertex type separately and then sums them up.
Consider for simplicity an online type $i$ with unit arrival rate.
Suppose that the fractional matching assigns an $x_i$ fraction of $i$ to \emph{its currently unmatched offline neighbors}.
If this is at the beginning, type $i$'s contribution would be mainly proportional to $x_i$, because the LP-based fractional matching serves as a reference that implicitly accounts for the likelihood that all $i$'s neighbors are already matched when a type-$i$ vertex comes.
As time proceeds, type $i$'s contribution to the potential would smoothly transition to be proportional to $\min \{ \frac{x_i}{\theta}, 1\}$ for some parameter $0 \le \theta \le \frac{1}{2}$.
This is because as we get closer to the end, it is less likely that all $i$'s neighbors would get matched within the remaining time.
In the extreme case when it is already at the end, type $i$'s contribution should simply be proportional to its unit arrival rate.
The mentioned functional form interpolates between the two extremes, and is inspired by the analysis of the \tophalf algorithm by \citet{HuangSY:STOC:2022}.
Finally, we optimize $\theta$ to get the $0.707$ competitive ratio.

\paragraph{Deterministic Online Stochastic Matching Algorithm.}
Due to the potential function analysis, we can further derandomize the algorithm by matching each online vertex to an unmatched offline neighbor for which the potential function decreases the least.
The algorithm-dependent part of the potential's decreased amount may be viewed as a regularization term, which measures how likely an offline neighbor could get matched later if we do not match it now.
We therefore call this algorithm \reggreedy.
By design, it is at least $0.707$-competitive for Online Stochastic Matching.
This is the first polynomial-time deterministic algorithm that breaks the $1-\frac{1}{e}$ barrier.%
\footnote{The simple greedy algorithm is $(1-\frac{1}{e})$-competitive \cite{GoelM:SODA:2008}. The folklore optimal online algorithm based on backward induction is also deterministic, but is computationally intractable.}
Empirical experiments further suggest that the derandomized \reggreedy is slightly better than the randomized \swor on real-world graphs.

\paragraph{Competitive Analysis of Balance SWOR.}
Last but not least, we offer the first non-trivial analysis of the \balanceswor algorithm, proving that it is at least $0.513$-competitive for Online Bipartite Matching.
Our analysis is based on the following observation which is simple in hindsight.
Although sampling without replacement fails to give any worst-case guarantee with respect to (w.r.t.) the online correlated selection problem~\cite{Gao:FOCS:2021}, in \balanceswor it only needs to handle sampling probabilities chosen by the \balance algorithm, which by design balances the total sampling probabilities assigned to different offline vertices.
This observation circumvents the worst-case instance for sampling without replacement.
As a result, we derive a restricted form of online correlated selection guarantees, and the stated competitive ratio for Online Bipartite Matching.

\subsection{Other Related Work}

After \citet{Karp:STOC:1990}, \citet{Aggarwal:SODA:2011} extended the \ranking algorithm and the $(1-\frac{1}{e})$ competitive ratio to vertex-weighted matching.
\citet{FeldmanMMM:FOCS:2009} proved that no online algorithm can achieve a non-trivial worst-case competitive ratio for edge-weighted matching.
They further proposed the \emph{free disposal model} and an optimal $(1-\frac{1}{e})$-competitive algorithm in this model under a large-market assumption.
AdWords is another noteworthy variant of Online Bipartite Matching, which was proposed by \citet*{Mehta:JACM:2007}.
They gave an optimal $(1-\frac{1}{e})$-competitive algorithm for AdWords also under a large-market assumption.
The algorithms by \citet{FeldmanMMM:FOCS:2009} and \citet{Mehta:JACM:2007} may be viewed as generalizations of the \balance algorithm.
Subsequently, these results have been streamlined and integrated within the online primal-dual framework \cite{Buchbinder:ESA:2007, Devanur:SODA:2013, Devanur:TEAC:2016}.
The aforementioned online correlated selection technique contributed to breaking the $\frac{1}{2}$ barrier in edge-weighted matching with free disposal~\cite{Blanc:FOCS:2022,Gao:FOCS:2021,Shin:ARXIV:2021,Fahrbach:2022:JACM} and AdWords~\cite{Huang:FOCS:2020}, without the large-market assumptions.

These extensions of Online Bipartite Matching have also been studied in the stochastic model.
\citet{HuangSY:STOC:2022} proposed \poissonocs and \tophalf which get the state-of-the-art $0.716$ and $0.706$ competitive ratios for vertex-weighted matching and edge-weighted matching with free disposal respectively.
For edge-weighted matching without free disposal, \citet{HuangSY:STOC:2022} proved a $0.703$ upper bound, separating it from the other settings.
\citet{Yan:arXiv:2022} recently proposed the \multisuggestedmatch algorithm which breaks the $1-\frac{1}{e}$ barrier, and this result was further improved by \citet{Feng:arXiv:2023}.
\citet{AouadM:EC:2023} studied edge-weighted online stochastic matching with correlated online vertex arrivals.
\citet{MaXX:WINE:2021} studied individual fairness maximization and group fairness maximization in online stochastic matching with integral arrival rate.

The literature has also studied the random order model of online matching which lies between the worst-case and stochastic models.
For (unweighted) Online Bipartite Matching with random arrival order, \citet*{Karande:STOC:2011} and \citet{Mahdian:STOC:2011} showed that \ranking's competitive ratio is between $0.696$ and $0.727$.
For vertex-weighted matching, \citet*{Huang:TALG:2019} and \citet{Jin:ARXIV:2020} proposed variants of \ranking that break the $1-\frac{1}{e}$ barrier.
For edge-weighted matching without free disposal, \citet*{Kesselheim:ESA:2013} generalized the secretary algorithm to get an optimal $\frac{1}{e}$ competitive ratio.
Finally, \citet{Devanur:EC:2009}, \citet*{KesselheimRTV:SICOMP:2018}, and \citet{AgrawalD:SODA:2014} proposed $1-\epsilon$ competitive algorithms for AdWords under a large-market assumption.

%% file: preliminaries.tex
\section{Preliminaries}
\label{sec:preliminaries}
\subsection{Online Bipartite Matching}

We first define the Online Bipartite Matching problem, which was proposed by \citet*{Karp:STOC:1990}.
Consider an undirected bipartite graph $G = (I, J, E)$, where $I$ and $J$ are the sets of left-hand-side and right-hand-side vertices respectively, and $E$ is the set of edges.
Define $I_j$ to be the set of online neighbors of an offline vertex $j \in J$, and define $J_i$ similarly for $i \in I$.
In the beginning, an online algorithm knows only the right-hand-side vertices, but not the left-hand-side vertices and edges.
Then, the left-hand-side vertices arrive one by one.
When a left-hand-side vertex $i \in I$ arrives, the algorithm observes its incident edges, and must immediately and irrevocably decide how to match $i$.
Based on the distinct nature of left-hand-side and right-hand-side vertices, we will refer to them as online and offline vertices respectively.

We consider maximizing the size of the matching.
Following the standard competitive analysis, we compare the size of an online algorithm's matching, taking expectation over the algorithm's randomness, with the size of the optimal matching in hindsight, i.e., the best achievable objective if we had complete information about the bipartite graph $G = (I, J, E)$. 
The competitive ratio of an online algorithm is the minimum of this ratio among all bipartite graphs and all arrival orders of online vertices.
An algorithm is $\Gamma$-competitive if its competitive ratio is at least $\Gamma$.

\subsection{Online Stochastic Matching}

Next we present the Online Stochastic Matching problem introduced by 
\citet*{FeldmanMMM:FOCS:2009}, and the Natural LP recently developed by \citet{HuangS:STOC:2021}.

Consider an undirected bipartite graph $G = (I, J, E)$, where $I$ is the set of \emph{online vertex types}, $J$ the set of offline vertices, and $E$ the set of edges.
Let $I_j$ and $J_i$ denote the neighborhoods of $j \in J$ and $i \in I$.
We will refer to $G$ as the \emph{type graph}.
Online vertices of each type $i \in I$ arrives from time $0$ to $1$ by a Poisson process with arrival rate $\lambda_i > 0$.%
\footnote{The original model by \citet{FeldmanMMM:FOCS:2009} considered $\Lambda = \sum_{i \in I} \lambda_i$ online vertices each of which samples its type $i$ with probability $\frac{\lambda_i}{\Lambda}$. \citet{HuangS:STOC:2021} argued that Poisson arrivals are more convenient for the independence of different online vertex types, and proved an asymptotic equivalence of the two models.}
An online algorithm knows the type graph and the arrival rates from the beginning, but not the realization of online vertex arrivals.
When an online vertex arrives, the algorithm sees its type $i$ and thus the incident edges, and must immediately and irrevocably decide how to match it.

We still consider maximizing the cardinality of the matching.
On the one hand, we have the expected size of an online algorithm's matching over the algorithm's randomness and the random online vertex arrivals.
On the other hand, we consider the expected size of the optimal matching in hindsight over the random online vertex arrivals, i.e., the best achievable objective in expectation if we always choose the optimal matching w.r.t.\ the realized bipartite graph.
The competitive ratio of an online algorithm is the minimum of this ratio among all possible type graphs and arrival rates.
An algorithm is $\Gamma$-competitive if its competitive ratio is at least $\Gamma$.

\paragraph{Natural Linear Program.}
Consider $x = (x_{ij})_{(i,j) \in E}$ where $x_{ij}$ corresponds to the probability that offline vertex $j$ is matched to an online vertex of type $i$.
We will use the following Natural LP:
\begin{equation}
	\label{eqn:natural-lp}
	\begin{aligned}
		\text{maximize} \quad & \textstyle \sum_{(i,j) \in E} x_{ij} \\[1ex]
		\text{subject to} \quad & \textstyle 
		\sum_{j \in J_i} x_{ij} \le \lambda_i && \forall i \in I \\[1ex]
		& \textstyle 
		\sum_{i \in S} x_{ij} \le 1 - \exp \big( - \sum_{i \in S} \lambda_i \big) && \forall j \in J, \forall S \subseteq I_j \\[1ex]
		& \textstyle x_{ij} \ge 0 && \forall (i,j)\in E
	\end{aligned}
\end{equation}

The Natural LP plays two roles in Online Stochastic Matching.
First, it is a computationally efficient upper bound of the optimal objective, because it can be solved in polynomial time using the ellipsoid method and a polynomial time separation oracle for its second constraint~\cite{HuangS:STOC:2021}.
Further, online algorithms can use the optimal solution of the Natural LP, which we will abuse notation and denote also $x = (x_{ij})_{(i,j) \in E}$, as a reference for making online decisions.

\begin{lemma}[Converse Jensen Inequality, c.f., \citet{HuangS:STOC:2021}]
	\label{lem:converse-jensen}
	For any convex $f : [0, 1] \to \R$ such that $f(0) = 0$, and any offline vertex $j \in J$:
	\[
	\sum_{i \in I} \lambda_i f\Big( \frac{x_{ij}}{\lambda_i} \Big) \le \int_0^{\infty} f\big(e^{-\lambda}\big) d\lambda
	~.
	\]
\end{lemma}

Denote function $\max \{x, 0\}$ as $x^+$.
For any $0 \le \theta \le \frac{1}{2}$, consider $f(x) = \frac{1}{\theta} \big(x - (1-\theta)\big)^+$.
We have the next corollary of the Converse Jensen Inequality.

\begin{corollary}
	\label{cor:converse-jensen-theta}
	For any $0 \le \theta \le \frac{1}{2}$ and any offline vertex $j \in J$:
	\[
	\frac{1}{\theta} \sum_{i \in I} \big(  x_{ij} - (1-\theta)\lambda_i \big)^+ \le 1 + \frac{1-\theta}{\theta} \ln(1-\theta)
	~.
	\]
\end{corollary}

%% file: SWOR.tex
\section{Stochastic Sampling Without Replacement}
\label{sec:stochasticSWOR}

This section presents the first nontrivial competitive analysis of the \swor algorithm.
We further derandomize it and get a deterministic algorithm called \reggreedy which achieves the same competitive ratio.
We first formally define the \swor algorithm.

\begin{tcolorbox}[beforeafter skip=10pt]
	\textbf{Stochastic Sampling Without Replacement}\\[1ex]
	\emph{Input at the beginning:}
	\begin{itemize}[itemsep=0pt, topsep=4pt]
		\item Type graph $G = (I, J, E)$;
		\item Arrival rates $(\lambda_i)_{i \in I}$;
		\item Solution of the Natural LP $(x_{ij})_{(i,j) \in E}$.
	\end{itemize}
	\smallskip
	\emph{When an online vertex of type $i \in I$ arrives at time $0 \le t \le 1$:}
	\begin{itemize}[itemsep=0pt, topsep=4pt]
		\item Match it to an unmatched offline neighbor $j$ with probability proportional to $x_{ij}$.
	\end{itemize}
\end{tcolorbox}

\smallskip

\swor is at least weakly better than \suggestedmatch, which samples an offline vertex $j$ regardless of whether it is already matched.
Hence, \swor is at least $1-\frac{1}{e} \approx 0.632$-competitive.
Our main result is a substantially better ratio.

\begin{theorem}
	\label{thm:swor}
	\swor is $0.707$-competitive for Online Stochastic Matching.
\end{theorem}

\subsection{Potential Function Analysis of Online Stochastic Matching Algorithms}
\label{subsec:analysis-framework}

\paragraph{Potential Function Analysis Framework.}
For any time $0 \le t \le 1$, let $\alg(t)$ denote the number of edges matched by the online algorithm.
We keep track of the unmatched portion of fractional matching $(x_{ij})_{(i,j)\in E}$ using the following notations:
\begin{equation}
	\label{eqn:swor-x}
	x_{ij}(t) =
	\begin{cases}
		x_{ij} & \mbox{if $j$ is unmatched} \\
		0 & \mbox{otherwise}
	\end{cases}
	~,\quad
	x_j(t) = \sum_{i \in I_j} x_{ij}(t)
	~,\quad
	x(t) = \big\{ x_{ij}(t) \big\}_{(i,j) \in E}
	~.
\end{equation}

Further, we define auxiliary variables for match rates, i.e., the values of $x_{ij}(t)$ divided by the arrival rate $\lambda_i$ of online vertex type $i$:
\begin{equation}
	\label{eqn:swor-rho}
	\rho_{ij}(t) = \frac{x_{ij}(t)}{\lambda_i}
	~,\quad
	\rho_i(t) = \sum_{j \in J_i} \rho_{ij}(t)
	~.
\end{equation}

We will design a potential function $\pot(t)$, which depends on $x(t)$, i.e., the unmatched portion of the fractional matching, and the current time $0 \le t \le 1$.
To prove that the competitive ratio is at least $\Gamma$, our potential function needs to satisfy that:
\begin{enumerate}[label=(\Alph*)]
	\item $\pot(0) \ge \Gamma \cdot \opt$~; \label{con:start}
	\item $\pot(1) = 0$~; and \label{con:end}
	\item $\E\, \big[ \alg(t) + \pot(t) \big]$ is non-decreasing in $0 \le t \le 1$. \label{con:monotone}
\end{enumerate}

Since $\alg(0) = 0$ by definition, combining these conditions gives $\E\,\alg(1) \ge \Phi(0) \ge \Gamma \cdot \opt$.

\begin{remark} 
	The value of this analysis framework lies beyond the online unweighted stochastic matching problem. It can also be extended to online edge-weighted stochastic matching with free disposal. Using this framework, we propose a deterministic online algorithm for this problem that achieves the same competitive ratio as the state-of-the-art \tophalf by \citet{HuangSY:STOC:2022}. See Appendix \ref{sec:extension} for details.
\end{remark}

\paragraph{Design of Potential Function.}
Next, for some parameter $0 \le \theta \le \frac{1}{2}$ to be determined, define function $p : [0, 1] \to [0, 1]$ as:
\[
p(\rho) = \min \Big\{ \frac{\rho}{\theta}, 1 \Big\} ~.
\]

Initially, each online type $i$'s contribution to our potential function is mainly proportional to $\lambda_i \rho_i(t)$.
In other words, we directly take the LP-based fractional matching as our reference, since that implicitly takes into account the possibility that all offline neighbors are already matched when a type-$i$ online vertex arrives.
As time proceeds, an online type $i$'s contribution will smoothly transition to be proportional to $\lambda_i p(\rho_i(t))$.
This is because with less time remaining, it is less likely that all offline neighbors would get matched after time $t$.
We can therefore expect an online type $i$'s contribution to be closer to its arrival rate $\lambda_i$ even if $\rho_i(t)$ is small.

To formalize this transition, we further define functions $\alpha, \beta: [0, 1] \to [0, 1]$ as follows:
\begin{align*}
	\alpha(t) & = 
	1-\frac{\frac{1}{\theta} e^{-(1-\ln(1-\theta))(1-t)} -(1-\ln(1-\theta)) e^{-\frac{1}{\theta}(1-t)}}{\frac{1}{\theta} -1+\ln(1-\theta)}    
	~, \\ %
	\beta(t) & = 
	\frac{e^{-(1-\ln(1-\theta))(1-t)} -e^{-\frac{1}{\theta}(1-t)}}{\frac{1}{\theta} -1+\ln(1-\theta)}
	~, %
\end{align*}
derived by solving the second-order differential equations in the next lemma.

Our potential function is:
\[
\pot(t) = \alpha(t) \sum_{(i,j) \in E} x_{ij}(t) + \beta(t) \sum_{i \in I} \lambda_i p \big( \rho_i(t) \big)
~.
\]

We next summarize the properties that characterize functions $\alpha(t)$ and $\beta(t)$ as a lemma, which we will use in our analysis.
The proofs are basic calculus and hence omitted.
Interested readers may also plot both functions for $0 \le t \le 1$, e.g., with $\theta = \frac{1}{2}$, to verify that $\alpha(t)$ is larger than $\beta(t)$ for small $t$, but $\beta(t)$ becomes the dominant term when $t$ tends to $1$.

\begin{lemma}
	\label{lem:alpha-beta}
	Functions $\alpha(t)$ and $\beta(t)$ satisfy the following properties:
	\begin{enumerate}
		\item $\alpha(0)+\beta(0)
		=1-\dfrac{\left(\frac{1}{\theta } -1\right) e^{-(1-\ln(1-\theta))} +\ln(1-\theta) e^{-\frac{1}{\theta }}}{\frac{1}{\theta} -1+\ln(1-\theta)}~.$
		\label{prop: start}
		\item $\alpha(1)=\beta(1)=0~.$
		\label{prop: end}
		\item $\alpha(t)+\frac{1}{\theta}\beta(t)\le 1$ for any $0\le t \le 1~.$
		\label{prop: sum to less than 1}
		\item $\frac{d}{dt} \alpha (t) = -\frac{1-\ln(1-\theta )}{\theta } \beta (t)~.$
		\label{prop: derivative of alpha}
		\item $\frac{d}{dt} \beta (t)
		= \alpha (t) +\left(\frac{1}{\theta} +1-\ln(1-\theta)\right) \beta (t) -1~.$
		\label{prop: derivative of beta}
	\end{enumerate}
\end{lemma}

\subsection{Proof of Theorem \ref{thm:swor}}

We next verify the aforementioned sufficient conditions for competitive ratio:
\begin{equation}
	\label{eqn:swor-ratio}
	\Gamma = 1-\frac{\left(\frac{1}{\theta } -1\right) e^{-(1-\ln(1-\theta))} +\ln(1-\theta) e^{-\frac{1}{\theta }}}{\frac{1}{\theta} -1+\ln(1-\theta)} ~.
\end{equation}
Theorem \ref{thm:swor} then follows by choosing $\theta = 0.4254$ which gives $\Gamma\approx 0.7078 > 0.707$.

\paragraph{Condition \ref{con:start}.}
Since $p(\rho) \ge \rho$ for any $0 \le \rho \le 1$, we have:
\[
\sum_{i \in I}  \lambda_i p \big( \rho_i(0) \big) \ge \sum_{i \in I} \lambda_i \rho_i(0) = \sum_{(i,j) \in E} x_{ij}(0)=\opt
~.
\]
Therefore, we get that:
\begin{align*}
	\pot(0) 
	= \alpha(0) \sum_{(i,j) \in E} x_{ij}(0) + \beta(0) \sum_{i \in I} \lambda_i p \big( \rho_i(0) \big)
	\ge \big(\alpha(0) + \beta(0)\big) \cdot \opt ~.
\end{align*}
Further by Property \ref{prop: start} of Lemma \ref{lem:alpha-beta} and by the definition of $\Gamma$ in Eqn.~\eqref{eqn:swor-ratio}, we have $\alpha(0) + \beta(0) = \Gamma$.

\paragraph{Condition \ref{con:end}.}
It follows by Property \ref{prop: end} of Lemma \ref{lem:alpha-beta}.

\paragraph{Condition \ref{con:monotone}.}
For any time $0\le t \le 1$, and conditioned on any set of matched offline vertices at time $t$ and the corresponding $x(t)$, we next show that the derivative of $\E\, [ \alg(t) + \pot(t) \mid x(t) ]$ is nonnegative.
The derivative can be expressed as:
\begin{equation}
	\label{eqn:SWOR-derivative-of-A-plus-Phi}
	\begin{multlined}
		\frac{d}{dt} \E\, \big[ \alg(t) \mid x(t) \big]
		+
		\alpha(t) \frac{d}{dt} \E\, \Big[ \sum_{(i,j) \in E} x_{ij}(t) \mid x(t) \Big]
		+
		\beta(t) \frac{d}{dt} \E\, \Big[ \sum_{i \in I} \lambda_i p\big(\rho_i(t)\big) \mid x(t) \Big]
		\\		
		+\Big( \frac{d}{dt} \alpha(t) \Big) \sum_{(i,j) \in E} x_{ij}(t)
		+
		\Big( \frac{d}{dt} \beta(t) \Big) \sum_{i \in I} \lambda_i p\big(\rho_i(t)\big) ~.
	\end{multlined}
\end{equation}

Define the match rate of edge $(i,j)$ at time $t$ as:
\begin{equation*}
	y_{ij}(t) = \lambda_i \cdot \frac{\rho_{ij}(t)}{\rho_i(t)}
	~.
\end{equation*}
where $\lambda_i$ is the arrival rate of $i$, and $\frac{\rho_{ij}(t)}{\rho_i(t)}$ is the probability of sampling $j$ conditioned on $i$'s arrival by the definition of sampling without replacement.

Further let $y_{j}(t)=\sum_{i\in I_j} y_{ij}(t)$ be the total match rate of offline vertex $j$ at time $t$.
The first line of Eqn.~\eqref{eqn:SWOR-derivative-of-A-plus-Phi} equals:
\begin{equation}
	\label{eqn:SWOR-derivative-first-line}
	\sum_{j \in J} y_{j}(t) \Big( 1 - \alpha(t) x_j(t) - \beta(t) \sum_{i \in I_j} \lambda_{i} \big( p(\rho_{i}(t)) - p(\rho_{i}(t) - \rho_{ij}(t))\big) \Big)
	~.
\end{equation}
The three terms inside the big parentheses correspond to the respective changes in
$\E\, \big[ \alg(t) \mid x(t) \big]$,
$\alpha(t) \E\, \big[ \sum_{(i,j) \in E} x_{ij}(t) \mid x(t) \big]$
and
$\beta(t) \E\, \big[ \sum_{i \in I} \lambda_i p\big(\rho_i(t)\big) \mid x(t) \big]$
when $j$ gets matched.

The coefficient of $y_j(t)$ in Eqn.~\eqref{eqn:SWOR-derivative-first-line} is at least $1-\alpha(t)-\frac{1}{\theta}\beta(t)$, because the derivative of function $p$ is at most $\frac{1}{\theta}$ and $\sum_{i \in I_j} \lambda_i \rho_{ij}(t) = x_j(t) \le 1$.
It is therefore nonnegative by Property \ref{prop: sum to less than 1} of Lemma \ref{lem:alpha-beta}.
Hence, replacing $y_{ij}(t)$ with $y_{ij}(t)p(\rho_i(t))$ makes Eqn.~\eqref{eqn:SWOR-derivative-of-A-plus-Phi} weakly smaller.
Further, since $x_j(t)\leq 1$, replacing $\alpha(t)x_j(t)$ with $\alpha(t)$ also makes Eqn.~\eqref{eqn:SWOR-derivative-of-A-plus-Phi} weakly smaller.
Hence, we conclude that:
\begin{align}
	\eqref{eqn:SWOR-derivative-of-A-plus-Phi}
	\ge
	&
	\big(1 - \alpha(t) \big)
	\sum_{i \in I}
	\lambda_i p\big(\rho_i(t)\big)
	-\beta(t)
	\sum_{j \in J} \Big(\sum_{i\in I_j} y_{ij}(t) p(\rho_i(t))\Big)
	\Big(\sum_{i \in I_j} \lambda_{i} \big( p(\rho_{i}(t)) - p(\rho_{i}(t) - \rho_{ij}(t))\big) \Big) \notag \\
	& \quad
	+\Big( \frac{d}{dt} \alpha(t) \Big) \sum_{(i,j) \in E} x_{ij}(t)
	+\Big( \frac{d}{dt} \beta(t) \Big) \sum_{i \in I}
	\lambda_i p\big(\rho_i(t)\big)
	~.
	\label{eqn:SWOR-derivative-of-A-plus-Phi-tighten}
\end{align}

If we treat functions $\alpha$, $\beta$, and their derivatives as coefficients, then the second term on the right-hand-side of Eqn.~\eqref{eqn:SWOR-derivative-of-A-plus-Phi-tighten} is quadratic, while the other terms are linear.
The next lemma bounds this term by a linear combination of two linear terms $\sum_{i\in I} \lambda_i p(\rho_i(t))$
and $\sum_{(i,j) \in E} x_{ij}(t)$ in this inequality.
Its proof is deferred to the end of the section.
\begin{lemma}
	\label{lem:SWOR-diff-ineq}
	For any $0\le t \le 1$, the following inequality holds:
	\begin{equation}
		\begin{multlined}
			\label{eqn:SWOR-diff-ineq}
			\sum_{j\in J} \Big(\sum_{i\in I_j} y_{ij}(t) p(\rho_i(t))\Big) \Big(\sum_{i \in I_j} \lambda_{i} \big( p(\rho_{i}(t)) - p(\rho_{i}(t) - \rho_{ij}(t))\big) \Big)\\
			\le \left(\frac{1}{\theta} +1-\ln(1-\theta)\right) \sum_{i\in I} \lambda_i p(\rho_i(t))
			-\frac{1-\ln(1-\theta)}{\theta}\sum_{(i,j) \in E} x_{ij}(t)~.
		\end{multlined}
	\end{equation}
\end{lemma}

We remark that a similar lemma holds if on the left we use the original match rates $y_{ij}(t)$ instead of the scaled $y_{ij}(t) p(\rho_i(t))$, but the first term on the right would be $\sum_{i \in I} \lambda_i$.
This is not an existing term in Eqn.~\eqref{eqn:SWOR-derivative-of-A-plus-Phi-tighten}, and therefore we would not have the cancellation in our subsequent analysis.
This is why we artificially scale the match rates in the previous step of our analysis.

Lemma~\ref{lem:SWOR-diff-ineq} is partly inspired by a corresponding lemma in \citet{HuangSY:STOC:2022}. %
The main differences are threefold.
Firstly and most importantly, our approach decouples the (lower bounds of) match rates of online vertices in the lemma with the actual match rates of the online algorithm.
Following the approach of \citet{HuangSY:STOC:2022}, the \emph{algorithm} would need to artificially decrease the match rate of each online vertex $i$ by a $p(\rho_i(t))$ factor, and we would not be able to analyze \swor.
Second, we distribute the match rate of each online vertex $i$ to its offline neighbors proportional to the fractional matching $x(t)$, while their approach considered a distorted distribution of the match rate, which once again disagrees with \swor.
Last but not least, our approach extends to any $0 \le \theta \le \frac{1}{2}$ while their approach only allows $\theta = \frac{1}{2}$.

Applying Lemma \ref{lem:SWOR-diff-ineq} to Eqn.~\eqref{eqn:SWOR-derivative-of-A-plus-Phi-tighten}, we have:
\begin{equation*}
	\begin{multlined}
		\eqref{eqn:SWOR-derivative-of-A-plus-Phi}
		\ge \left( 1 - \alpha(t)
		- \beta(t) \Big(\frac{1}{\theta} +1-\ln(1-\theta)\Big)
		+\frac{d}{dt} \beta (t)\right)
		\sum_{i\in I} \lambda_i p(\rho_i(t))  \\
		+ \left(\beta(t) \frac{1-\ln(1-\theta)}{\theta} +\frac{d}{dt} \alpha (t) \right)\sum_{(i,j) \in E} x_{ij}(t)
		~.
	\end{multlined}
\end{equation*}

This equals zero by Properties \ref{prop: derivative of alpha} and \ref{prop: derivative of beta} of Lemma \ref{lem:alpha-beta}.
Hence, we complete the verification of Condition~\ref{con:monotone} and thus the proof of Theorem~\ref{thm:swor}.

\subsection{Regularized Greedy}

This subsection derandomizes \swor and introduces a deterministic algorithm called \reggreedy which is also $0.707$-competitive.

Observe that Conditions~\ref{con:start} and \ref{con:end} always hold, independent of the online algorithm's decisions.
Therefore, the competitive analysis from the last subsection applies to any algorithm that satisfies Condition~\ref{con:monotone}, i.e., the monotonicity of $\E\,\big[ \alg(t) + \pot(t) \big]$ for $0 \le t \le 1$.
Conditioned on the set of matched offline vertices and the corresponding $x(t)$ at time $t$, it suffices to ensure that the derivative $\frac{d}{dt} \E\, \big[ \alg(t) + \pot(t) \mid x(t) \big]$ is nonnegative, which can be written as follows:
\begin{equation}
	\label{eqn:reg-greedy-derivative}
	\begin{multlined}
		\overbrace{\vphantom{\Bigg|} \frac{d}{dt} \E\, \big[ \alg(t) \mid x(t) \big]
			+
			\alpha(t) \frac{d}{dt} \E\, \Big[ \sum_{(i,j) \in E} x_{ij}(t) \mid x(t) \Big]
			+
			\beta(t) \frac{d}{dt} \E\, \Big[ \sum_{i \in I} \lambda_i p\big(\rho_i(t)\big) \mid x(t) \Big]}^{\mbox{\small depends on algorithm's decisions}} \\[1ex]
		+
		\underbrace{\Big( \frac{d}{dt} \alpha(t) \Big) \sum_{(i,j) \in E} x_{ij}(t)
			+
			\Big( \frac{d}{dt} \beta(t) \Big) \sum_{i \in I} \lambda_i p\big(\rho_i(t)\big)}_{\mbox{\small independent of algorithm's decisions}}
		~.
	\end{multlined}
\end{equation}

Given any online algorithm $\alg$, we can write the first line above, i.e., the part that depends on the algorithm's decisions, as follows:
\[
\sum_{i \in I} \lambda_i \sum_{j \in J_i} \Pr \big[ \mbox{$\alg$ matches $i$ to $j$} \mid x(t) \big] \bigg( 1 - \underbrace{\alpha(t) x_j(t) - \beta(t) \sum_{i' \in I_j} \lambda_{i'} \Big( p \big(\rho_{i'}(t)\big) - p\big(\rho_{i'}(t) - \rho_{i'j}(t)\big) \Big)}_{\mbox{\small regularization term}} \bigg)~.
\]

Intuitively, this regularization term is our estimated likelihood that offline vertex $j$ would be matched in the remaining time, and thus the cost of matching it right now.

The analysis of \swor may be viewed as a constructive proof of the \emph{existence of online algorithms} which ensure nonnegativity of the derivative in Eqn.~\eqref{eqn:reg-greedy-derivative}.
Instead, we may directly maximize the first line of Eqn.~\eqref{eqn:reg-greedy-derivative}:
conditioned on the arrival of an online vertex of type $i$ at time $t$, match it to an unmatched offline neighbor $j$ with the smallest regularization term, i.e.:
\begin{equation}
	\label{eqn:regularization}
	\alpha(t) x_j(t) + \beta(t) \sum_{i' \in I_j} \lambda_{i'} \Big( p\big(\rho_{i'}(t) \big) - p \big( \rho_{i'}(t) - \rho_{i'j}(t) \big) \Big)
	~.
\end{equation}

\begin{tcolorbox}[beforeafter skip=10pt]
	\textbf{Regularized Greedy}\\[1ex]
	\emph{Input at the beginning:}
	\begin{itemize}[itemsep=0pt, topsep=4pt]
		\item Type graph $G = (I, J, E)$;
		\item Arrival rates $(\lambda_i)_{i \in I}$;
		\item Solution $(x_{ij})_{(i,j) \in E}$ of the Natural LP.
	\end{itemize}
	\smallskip
	\emph{When an online vertex of type $i \in I$ arrives at time $0 \le t \le 1$:}
	\begin{itemize}[itemsep=0pt, topsep=4pt]
		\item Match it to an unmatched offline neighbor $j$ that minimizes Eqn.~\eqref{eqn:regularization}.
	\end{itemize}
\end{tcolorbox}

As a corollary of the definition of \reggreedy and the analysis of \swor, we have the next theorem which makes \reggreedy the first polynomial-time deterministic algorithm that breaks the $1-\frac{1}{e}$ barrier in Online Stochastic Matching.

\begin{theorem}
	\label{thm:regularized-greedy}
	\reggreedy is $0.707$-competitive for Online Stochastic Matching.
\end{theorem}

\subsection{Proof of Lemma \ref{lem:SWOR-diff-ineq}}
For notational simplicity, we omit $t$ in the variables in this proof.
We first restate Eqn.~\eqref{eqn:SWOR-diff-ineq} below:
\begin{equation*}
	\begin{multlined}
		\sum_{j\in J} \Big(\sum_{i\in I_j} y_{ij} p(\rho_i)\Big) \Big(\sum_{i \in I_j} \lambda_{i} \big( p(\rho_{i}) - p(\rho_{i} - \rho_{ij})\big) \Big)\\
		\le \left(\frac{1}{\theta} +1-\ln(1-\theta)\right) \sum_{i\in I} \lambda_i p(\rho_i)
		-\frac{1-\ln(1-\theta)}{\theta}\sum_{(i,j) \in E} x_{ij}~.
	\end{multlined}
\end{equation*}
where $\rho_i = \sum_{j \in J_i} \rho_{ij}$, $x_{ij} = \lambda_i \rho_{ij}$ and $y_{ij} = \lambda_i \frac{\rho_{ij}}{\rho_i}$.

\bigskip

We will prove a slightly stronger claim that the above inequality holds as long as:
\begin{enumerate}
	\item For any $i \in I$, $0\leq \rho_i\leq 1$;
	\item For any $j \in J$, $\sum_{i\in I_j} x_{ij}\le 1$; and
	\item For any $j \in J$, $\frac{1}{\theta} \sum_{i \in I_j} \big(x_{ij} - (1-\theta)\lambda_i \big)^+ \le 1 + \frac{1-\theta}{\theta} \ln(1-\theta)$.
\end{enumerate}

These three relaxed constraints are satisfied by our variables.
The first two constraints follow by the definitions of the variables.
The third constraint holds because of Corollary~\ref{cor:converse-jensen-theta}.
They are a relaxation because they do not imply, e.g., the second constraint of the Natural LP.

Under these relaxed constraints, it suffices to prove the inequality under three assumptions:
\begin{enumerate}[label=(\Roman*)]
	\item For any $i \in I$, $\rho_i\ge \theta$.
	\label{asm:rho<theta}
	\item For any $i \in I$ with $\theta\le \rho_i\le 1-\theta$, there is an offline neighbor $j \in J_i$ such that $\rho_{ij}=\rho_i$.		\label{asm:theta<=rho<=1-theta}
	\item For any $i \in I$ with $1-\theta < \rho_i \le 1$, there is an offline neighbor $j \in J_i$ such that  $\rho_{ij}\geq 1-\theta$.
	\label{asm:rho>1-theta}
\end{enumerate}

Concretely, we next prove that for any case that violates one of these assumptions, we can construct a new case that fixes the violation while tightening the inequality subject to the above relaxed constraints.
After that we will prove the inequality under these assumptions.
\paragraph{Assumption \ref{asm:rho<theta}.}
Suppose that an online vertex type $i \in I$ has $\rho_{i} < \theta$.
We replace $i$ with a new online vertex type $i'$ with the same set of offline neighbors, where:
\[
\lambda_{i'}=\frac{\lambda_i \rho_i}{\theta}
~,\quad
\rho_{i'}=\theta
~,\quad
\rho_{i'j}=\frac{\theta\rho_{ij}}{\rho_i}
~.
\]

For any offline neighbor $j$, the above parameters ensure that:
\[
x_{i'j}=x_{ij}
~,\quad
y_{i'j}p(\rho_{i'})=y_{ij}p(\rho_i)
~,\quad
\lambda_{i'} \big( p(\rho_{i'}) - p(\rho_{i'} - \rho_{i'j})\big)=\lambda_{i} \big( p(\rho_{i}) - p(\rho_{i} - \rho_{ij})\big)
~.
\]

Hence, Eqn.~\eqref{eqn:SWOR-diff-ineq} stays the same, and the first two constraints still hold.
Finally, the third constraint also holds because both $i$ and $i'$ contribute zero to its left-hand-side.

We remark that under Assumption \ref{asm:rho<theta}, Eqn.~\eqref{eqn:SWOR-diff-ineq} simplifies to:
\begin{equation}
	\label{eqn:SWOR-diff-ineq-restate-with-asm1}
	\sum_{j\in J} y_{j}\sum_{i\in I_j} \lambda_{i}\big(p(\rho_{i}) 
	-p(\rho_{i} -\rho_{ij})\big) 
	\le \left(\frac{1}{\theta} +1-\ln(1-\theta)\right) 
	\sum_{j\in J} y_{j} 
	- \frac{1-\ln(1-\theta)}{\theta}\sum_{j\in J} x_{j}
	~.
\end{equation}

\paragraph{Assumption \ref{asm:theta<=rho<=1-theta}.}
Suppose that an online vertex type $i \in I$ has $\rho_i\le 1-\theta$ and at least two offline neighbors $j$ whose $\rho_{ij}>0$.
We decompose $i$ into finitely many new online vertex types $i_j$, one for each offline vertex $j$ whose $\rho_{ij}>0$, where:
\[
\lambda_{i_j} = \frac{\lambda_i \rho_{ij}}{\rho_{i}}
~,\quad 
\rho_{i_j}=\rho_{{i_j}j}=\rho_{i}
~.
\]

For any offline neighbor $j$ with $\rho_{ij} > 0$, the above parameters ensure that $x_{i_j} = x_{i_j j} = x_{ij}$ and $y_{i_j} = y_{i_j j} = y_{ij}$.
Further by the concavity of $p$, we have:
\begin{equation*}
	\lambda_{i_j}\big(p(\rho_{i_j}) 
	-p(\rho_{i_j} -\rho_{{i_j}j})\big)
	=\frac{\lambda_{i} \rho_{ij}}{\rho_{i}} p(\rho_{i}) \ge \lambda_{i}\big( p(\rho_{i}) -p(\rho_{i} -\rho_{ij})\big).
\end{equation*}

Hence, our decomposition tightens Eqn.~\eqref{eqn:SWOR-diff-ineq-restate-with-asm1}, and still satisfies the first two constraints.
Finally, the third constraint also continues to hold because for any $j$, both $i$ and $i_j$ contribute zero to the left-hand-side of the constraint.

\paragraph{Assumption \ref{asm:rho>1-theta}.}
Suppose that an online vertex type $i \in I$ has $\rho_i > 1-\theta$, but $\rho_{ij} < 1-\theta$ for all offline neighbors $j \in J_i$.
Let $S = \{ j \in J_i : \rho_{ij} > 0 \}$ denote the offline neighbors that are fractionally matched to $i$.
We decompose $i$ into $|S|(|S|-1)$ new online vertex types $(j,k)$, one for each ordered pair $(j,k)$ from $S$ with $j\ne k$,
where:
\[
\rho_{(j,k)}=\rho_i
~,\quad
\rho_{(j,k),j} = 1-\theta
~,\quad
\rho_{(j,k),k}= \rho_i-1+\theta
~.
\]

It remains to choose arrival rates $\lambda_{(j,k)}$'s such that the decomposition tightens the inequality subject to the three constraints.
This reduces to an LP feasibility problem, which we summarize it as the next lemma but defer its proof to Appendix~\ref{app:stochasticSWOR}.

\begin{lemma}
	\label{lem:SWOR-reduc-feasibility-mass-distri}
	There are $\lambda_{(j,k)} \ge 0$ for $j, k \in S$, $j \ne k$, such that 
	for any $j \in S$ we have:
	\begin{equation}
		\label{lem:mass-distri-1}
		\sum_{k\in S\colon k\ne j} \lambda_{(j,k)} (1-\theta)+
		\sum_{k\in S\colon k\ne j} \lambda_{(k,j)} (\rho_i-1+\theta)
		= \lambda_i \rho_{ij}
		~,
	\end{equation}
	and:
	\begin{equation}
		\label{lem:mass-distri-2}
		\sum_{k \in S \colon k \ne j} \lambda_{(j,k)} \big( 1 - p(\rho_i - 1 + \theta) \big) \ge \lambda_i \big( p(\rho_i) - p(\rho_i - \rho_{ij}) \big)
		~.
	\end{equation}
\end{lemma}

By our choice of parameters for the new types, Eqn.~\eqref{lem:mass-distri-1} is equivalent to:
\[
\sum_{k\in S\colon k\ne j} x_{(j,k)j} +
\sum_{k\in S\colon k\ne j} x_{(k,j)j}
= x_{ij}
~.
\]

That is, the new types' combined contribution to $x_j$ equals $x_{ij}$.
This also means that the new types' combined contribution to $y_j$ equals $y_{ij}$, since the contributions to $x_j$ and $y_j$ differ by a fixed $\rho_i$ factor.

Combining this with Eqn.~\eqref{lem:mass-distri-2} shows that the decomposition tightens the inequality while satisfying the first two constraints.
Finally the third constraint still holds because both $i$ and the new types contribute zero to its left-hand-side.

\paragraph{Proof of Lemma~\ref{lem:SWOR-diff-ineq} Under Assumptions.}
Recall that it is sufficient to consider the simplified inequality in Eqn.~\eqref{eqn:SWOR-diff-ineq-restate-with-asm1}, which we rewrite below:
\[
\sum_{j\in J} y_{j}\sum_{i\in I_j} \lambda_{i}\big(p(\rho_{i}) 
-p(\rho_{i} -\rho_{ij})\big) 
\le \left(1+\frac{1-\theta}{\theta}\ln(1-\theta)\right)\sum_{j\in J} y_j
+ \frac{1-\ln(1-\theta)}{\theta}\sum_{j\in J} \big( y_j - x_j \big)
~.
\]

We would like to prove it for each offline vertex $j$ separately, but unfortunately the resulting inequality does not hold in general.
Instead, we will first amortize the last term.

For any offline vertex $j$, let $S_j$ be the set of online vertex types that mostly matches to $j$.
That is, $S_j$ is the set of online vertex types $i$ such that either $\rho_i =\rho _{ij} > 0$ or $\rho_{ij} \ge 1-\theta$.
Let $T_j$ be the set of other online vertex types that fractionally matches to $j$.
By our assumptions, any $i \in T_j$ satisfies $\rho_i > 1-\theta $ and $0< \rho_{ij} \le \rho_{i}-1+\theta$.
We have:
\[
\sum_{j\in J} \big( y_j - x_j \big)
= \sum_{i \in I} \sum_{j \in J_i} \big( y_{ij} - x_{ij} \big)
= \sum_{i \in I} \lambda_i (1 - \rho_i)
= \sum_{j\in J} \sum_{i \in S_j} \lambda_i (1 - \rho_i)
~.
\]

In other words, we attribute all contribution of an online vertex type $i$ to the offline neighbor $j$ that type $i$ mainly matches to.
With this amortization, it remains to prove for each offline vertex $j \in J$ that:
\begin{equation}
	\label{eqn:SWOR-amort}
	y_j\sum _{i\in I_j} \lambda_i\big( p(\rho_i) -p(\rho_i -\rho_{ij})\big)
	\le \left(1+\frac{1-\theta}{\theta}\ln(1-\theta)\right) y_j +\frac{1-\ln( 1-\theta)}{\theta }\sum_{i\in S_j} \lambda_i(1-\rho_i)
	~.
\end{equation}

We next define an auxiliary variable:
\begin{align*}
	\Delta_j
	&
	= 1+\frac{1-\theta}{\theta}\ln(1-\theta) -\frac{1}{\theta}\sum_{i\in S_j}\big(x_{ij} -(1-\theta) \lambda_i\big) \\
	&
	\ge 1+\frac{1-\theta}{\theta}\ln(1-\theta) -\frac{1}{\theta}\sum_{i\in I_j}\big(x_{ij} -(1-\theta) \lambda_i\big)^+ \ge 0
	~.
\end{align*}

For any $i \in I_j$, $p(\rho_i)-p(\rho_i-\rho_{ij})$ equals $\frac{1}{\theta} \big( \theta -\rho_i +\rho_{ij} \big)$ if $i \in S_j$, and is zero if $i \in T_j$.
Hence:
\begin{align*}
	\sum_{i\in I_j} \lambda_i\big( p( \rho_i) -p( \rho_i -\rho_{ij})\big) 
	&=\frac{1}{\theta}\sum_{i\in S_j} \lambda_i( \theta -\rho_i +\rho_{ij})\\
	&=\frac{1}{\theta}\sum_{i\in S_j} \lambda_i( 1-\rho_i) +1+\frac{1-\theta}{\theta}\ln( 1-\theta ) -\Delta_j~,
\end{align*}

Putting it into Eqn.~\eqref{eqn:SWOR-amort}, and merging terms that involve $\sum_{i \in S_j} \lambda_i(1-\rho_i)
$, it becomes:
\[
\sum_{i \in S_j} \lambda_i(1-\rho_i) \left( \frac{1}{\theta} y_j - \frac{1-\ln(1-\theta)}{\theta} \right) \le y_j \Delta_j
~.
\]

Hence, we only need to verify the following two inequalities:
\begin{align}
	\label{eqn:SWOR-amort-1}
	\sum_{i\in S_j} \lambda_i (1-\rho_i) & \le (1-\theta)y_j ~, \\
	\label{eqn:SWOR-amort-2}
	\frac{1}{\theta} y_j - \frac{1-\ln(1-\theta)}{\theta} & \le \frac{1}{1-\theta} \Delta_j ~.
\end{align}

\paragraph{Proof of Eqn.~\eqref{eqn:SWOR-amort-1}.}
It suffices to show for any $i\in S$ that:
\[
\lambda_i (1-\rho_i)\le (1-\theta) y_{ij} = (1-\theta) \frac{\lambda_i\rho_{ij}}{\rho_i}
~.
\]

If $\rho_i = \rho_{ij}$, it reduces to $\rho_i\ge \theta$, which holds by our assumptions.

If $\rho_{ij}\ge 1-\theta$, it reduces to $(1-\rho_i)\le \nicefrac{
	(1-\theta)^2}{\rho_i}$.
This holds because for any $\rho_i \ge 1 - \theta$ and any $0 \le \theta \le \frac{1}{2}$, we have $\rho_i (1-\rho_i) \le \theta(1-\theta) \le (1-\theta)^2$.

\paragraph{Proof of Eqn.~\eqref{eqn:SWOR-amort-2}.}
We can upper bound $y_j$ as follows:
\begin{equation*}
	\begin{aligned}
		y_j
		&
		= \sum_{i\in S_j} y_{ij} +\sum_{i\in T_j} y_{ij} \\
		&
		\le \sum_{i\in S_j} \lambda_{i} + \frac{1}{1-\theta}\sum_{i\in T_j} x_{ij}
		&& \mbox{($\rho_i \ge 1-\theta$ for $i \in T_j$)}
		\\
		&
		= \frac{1}{1-\theta} \bigg( \sum_{i\in S_j} x_{ij} - \sum_{i\in S_j} \big(x_{ij} -(1-\theta) \lambda_i\big) + \sum_{i\in T_j} x_{ij} \bigg)\\ 
		&
		= \frac{1}{1-\theta} x_j - \frac{\theta}{1-\theta} \Big(1+\frac{1-\theta}{\theta}\ln(1-\theta) - \Delta_j \Big) \\[2ex]
		&
		\le 1 - \ln(1-\theta) + \frac{\theta}{1-\theta} \Delta_j 
		~.
		&& \mbox{($x_j \le 1$)}
	\end{aligned}
\end{equation*}

%% file: balanceSWOR.tex
\section{Balance Sampling Without Replacement}
\label{sec:balanceSWOR}

This section gives the first non-trivial competitive analysis of the \balanceswor algorithm.
We assume without loss of generality that $I = \{1, 2, \dots, n\}$ for some positive integer $n$, where $i \in I$ is the $i$-th online vertex by their arrival order.

\balanceswor consists of two parts.
First, it runs an unbounded version of the \balance algorithm to maintain a fractional allocation of online vertices to offline vertices.
When each online vertex $i \in I$ arrives, the unbounded \balance algorithm allocates an $x_{ij}$ amount of vertex $i$ to each neighbor $j \in J_i$ such that (1) in total one unit of vertex $i$ is allocated to its neighbors, i.e., $\sum_{j \in J_i} x_{ij} = 1$, (2) every neighbor with $x_{ij} > 0$ is allocated with the same total amount of online vertices (including the allocation of $i$), i.e., has the same $\sum_{i' \le i} x_{i'j}$, and (3) every neighbor with $x_{ij} = 0$ must have already been allocated with a weakly larger total amount of online vertices.

Next, it resorts to sampling without replacement to find an integral matching, using the fractional allocation given by unbounded \balance as a reference.
That is, it matches online vertex $i$ to an unmatched offline neighbor $j$ with probability proportional to $x_{ij}$.

\begin{tcolorbox}[beforeafter skip=10pt]
	\label{alg:balanceswor}
	\balanceswor\\[1ex]
	\emph{State variable: (for every offline vertex $j \in J$)}
	\begin{itemize}[itemsep=0pt, topsep=4pt]
		\item $y_j^i = \sum_{i' \le i} x_{ij}$, total amount of online vertices $1$ to $i$ that is allocated to $j$; let $y_j^0 = 0$.
	\end{itemize}
	\smallskip
	\emph{When an online vertex of  $i \in I$ arrives:}
	\begin{itemize}[itemsep=0pt, topsep=4pt]
		\item Find threshold $\bar{y} > 0$ such that $\sum_{(i, j) \in E} (\bar{y} - y_j^{i-1})^{+}=1$.
		
		\item Let $x_{ij} = (\bar{y} - y_j^{i-1})^{+}$ for every neighbor $j$.
		
		\item Match $i$ to an unmatched neighbor $j$ with probability proportional to $x_{ij}$.
	\end{itemize}
\end{tcolorbox}

Trivially, \balanceswor is at least $0.5$-competitive because it is a greedy algorithm that always matches an online vertex whenever there is an unmatched neighbor.
The main result of this section is the first competitive analysis of \balanceswor that is strictly better than $0.5$.

\begin{theorem} 
	\label{thm:balanceswor}
	\balanceswor is at least $0.513$-competitive for Online Bipartite Matching.
\end{theorem}

 \paragraph{Comparison with Balance OCS.}
 \citet{Gao:FOCS:2021} proposed a similar algorithm called \balanceocs, based on their improved online correlated selection (OCS) technique.
 Their algorithm also runs the unbounded \balance algorithm to obtain a fractional allocation $x = (x_{ij})_{(i,j) \in E}$.
 Instead of sampling an unmatched neighbor $j$ with probability proportional to $x_{ij}$ like \balanceswor, their algorithm samples with probability proportional to $w(y_j^{i-1}) \cdot x_{ij}$ with function:
 \[
     w(y) = \exp \Big( y+\frac{1}{2}y^2+\frac{4-2\sqrt{3}}{3} y^3 \Big)
     ~.
 \]

 Although the adjustment to the sampling probabilities made by \citet{Gao:FOCS:2021} is necessary for online correlated selection, our empirical study suggests that it is redundant, if not harmful, in the context of Online Bipartite Matching.
 Theorem~\ref{thm:balanceswor} is the first step towards a theoretical explanation to this empirical observation.
 
 \subsection{Proof of Theorem~\ref{thm:balanceswor}}
 
 We first show that the competitive ratio of \balanceswor depends on the relation between the probability that an offline vertex $j$ is (un)matched and its fractional matched amount $y_j = y_j^n$.
 
 \begin{lemma}
 	\label{lem:ocs-to-ratio}
 	Suppose that for any instance of Online Bipartite Matching and any offline vertex $j$:
 	\[
 	\Pr \big[ \,\mbox{\rm vertex $j$ is unmatched}\, \big] \le q(y_j)
 	~,
 	\]
 	for some non-increasing, convex, and differentiable function $q : [0, \infty) \to [0, 1]$ with $q(0) = 1$.
 	Then, the competitive ratio of $\balanceswor$ is at least:
 	\[
 	\int_{0}^\infty e^{-z} \big( 1 - q(z) \big) dz
 	~.
 	\]
 \end{lemma}
 
 \citet{Gao:FOCS:2021} proved a similar lemma that requires the rounding procedure to be an online correlated selection algorithm, which informally means that the stated relation must be true even when the underlying fractional matching algorithm is arbitrary.
 By contrast, Lemma~\ref{lem:ocs-to-ratio} only needs the relation to hold when the underlying fractional matching algorithm is unbounded \balance.
 The proof is essentially the same, which we include in Appendix~\ref{app:balanceSWOR} for completeness.
 
 Note that the stated relation is trivially true for $q(y) = e^{-y}$, yet it only gives the trivial $0.5$ competitive ratio.
 The next lemma proves the relation for a function $q$ that decreases faster than $e^{-y}$ when $y > 1$.
 
 \begin{lemma}
 	\label{lem:balance-swor-ocs}
 	For any instance of Online Bipartite Matching and any offline vertex $j$:
 	\[
 	\Pr \big[ \,\mbox{\rm vertex $j$ is unmatched}\, \big] \le q(y_j)
 	~,
 	\]
 	where:
 	\begin{equation}
 		\label{eqn:q}
 		q(y) =
 		\begin{cases}
 			e^{-y} & 0 \le y \le 1~; \\
 			e^{-e^{y-1}} & y > 1~.
 		\end{cases}
 	\end{equation}
 \end{lemma}
 
 We will present its proof in the next subsection.
 The intuition is simple.
 Suppose that an offline vertex $j$ has already been fractionally matched to $y_j > 1$. 
 For any subsequent online vertex $i$ that is partly allocated to $j$ by unbounded \balance, we can conclude that any offline vertex $k \ne j$ that also gets a positive amount of $i$ must have $y_k > 0$.
 Therefore, there is nonzero chance that $k$ is already matched when $i$ arrives, which intuitively shall imply a larger probability that sampling without replacement would sample $j$ (if it is not yet matched).
 
 Theorem~\ref{thm:balanceswor} follows as a corollary of Lemmas~\ref{lem:ocs-to-ratio} and \ref{lem:balance-swor-ocs}, and a numerical calculation that:
 \[
 \Gamma = \int_0^\infty e^{-z} \big( 1 - q(z) \big) dz \approx 0.51304 > 0.513
 ~.
 \]
 
 Note that function $q$ in Eqn.~\eqref{eqn:q} is differentiable with derivative $\frac{1}{e}$ at $y = 1$.
 
 \subsection{Proof of Lemma~\ref{lem:balance-swor-ocs}}
 
 Throughout this subsection, $q$ refers to the function defined in Eqn.~\eqref{eqn:q}.
 We first state the following property about function $q$, whose proof is deferred to Appendix~\ref{app:balanceSWOR}.
 
 \begin{lemma}
 	\label{lem:balanceswor-condition}
 	For any $0 \leq z \leq 1$ and any $y \geq 0 $:
 	\[
 	\frac{z q\big( (y-z)^+ \big)}{1-z+z q\big( (y-z)^+ \big)} \leq \exp\Big( (1-z)\frac{q'(y)}{q(y)}\Big)
 	~.
 	\]
 \end{lemma}
 
 We next prove a strengthened version of Lemma~\ref{lem:balance-swor-ocs} by induction.
 The original Lemma~\ref{lem:balance-swor-ocs} is a special case of Lemma~\ref{lem:balance-swor-ocs-strengthened} when $i = n$ and $S$ is a singleton $\{j\}$.
 
 \begin{lemma}
 	\label{lem:balance-swor-ocs-strengthened}
 	For any instance of Online Bipartite Matching, any \emph{subset of offline vertices} $S$, and any online vertex $i \in I$:
 	\[
 	\Pr \big[ \,\mbox{\rm no vertex in $S$ is matched to online vertices $1$ to $i$}\, \big] \le \prod_{j \in S} q(y^i_j)
 	~.
 	\]
 \end{lemma}
 
 \begin{proof}
 	The first half of our proof follows the same framework as the online correlated selection analysis by \citet{Gao:FOCS:2021}.
 	For any $i \in I$, let $U^i$ denote the subset of offline vertices that are \emph{not} matched to online vertices $1$ to $i$;
 	define $U^0 = J$.
 	The lemma can be restated as follows.
 	We will prove by induction on $0 \le i \le n$ that for any subset of offline vertices $S$:
 	\begin{equation}
 		\label{eqn:ocs-induction}
 		\Pr \big[ S \subseteq U^i \big] \le \prod_{j \in S} q(y^i_j)
 		~.
 	\end{equation}
 	
 	The base case when $i = 0$ holds trivially since both sides equal $1$.
 	
 	Next suppose that Eqn.~\eqref{eqn:ocs-induction} holds for $i-1$.
 	Consider this inequality for $i$, and any subset of offline vertices $S$.
 	For any offline vertex $j \in J$, let $\bar{X}_j$ be the indicator that $j$ is not matched to vertices $1$ to $i-1$.
 	Further, for any subset of offline vertices $T \subseteq J$, let $\bar{X}_T = \prod_{j \in T} \bar{X}_j$ be the indicator that no vertex in $T$ is matched to online vertices $1$ to $i-1$.
 	The probability of $S \subseteq U^i$ is:
 	\begin{align*}
 		\Pr \big[ S \subseteq U^i \big]
 		&
 		= \E \bigg[ \bar{X}_S \bigg( 1 - \frac{\sum_{j \in S} x_{ij} \bar{X}_j}{\sum_{j \in J} x_{ij} \bar{X}_j} \bigg) \bigg] \\
 		&
 		= \E \bigg[ \bar{X}_S \bigg( 1 - \frac{\sum_{j \in S} x_{ij} \bar{X}_S}{\sum_{j \in S} x_{ij} \bar{X}_S + \sum_{j \notin S} x_{ij} \bar{X}_{S \cup \{j\}}} \bigg) \bigg] \\
 		&
 		= \frac{1}{\sum_{j \in S} x_{ij}} \E \bigg[ \frac{\big( \sum_{j \in S} x_{ij} \bar{X}_S \big) \big( \sum_{j \notin S} x_{ij} \bar{X}_{S \cup \{j\}} \big)}{\sum_{j \in S} x_{ij} \bar{X}_S + \sum_{j \notin S} x_{ij} \bar{X}_{S \cup \{j\}}} \bigg]
 		~,
 	\end{align*}
 	where the second equality holds because $\bar{X}_S \bar{X}_j = \bar{X}_S$ for any $j \in S$.
 	
 	By applying Jensen's inequality to concave function $\frac{xy}{x+y}$, and using the inductive hypothesis that $\E\, \bar{X}_T \le \prod_{j \in T} q(y^{i-1}_j)$, we conclude that:
 	\begin{align*}
 		\Pr \big[ S \subseteq U^i \big]
 		&
 		\le \frac{1}{\sum_{j \in S} x_{ij}} \cdot \frac{\big( \sum_{j \in S} x_{ij} \prod_{k \in S} q(y^{i-1}_k) \big) \big( \sum_{j \notin S} x_{ij} \prod_{k \in S \cup \{j\}} q(y^{i-1}_k) \big)}{\sum_{j \in S} x_{ij} \prod_{k \in S} q(y^{i-1}_k) + \sum_{j \notin S} x_{ij} \prod_{k \in S \cup \{j\}} q(y^{i-1}_k)} \\
 		&
 		= \prod_{k \in S} q(y^{i-1}_k) \cdot \frac{\sum_{j \notin S} x_{ij} q(y^{i-1}_j)}{\sum_{j \in S} x_{ij} + \sum_{j \notin S} x_{ij} q(y^{i-1}_j)} ~.
 	\end{align*}
 	
 	We remark that this is where our analysis diverges from the online correlated selection analysis by \citet{Gao:FOCS:2021}.
 	Their online correlated selection algorithm modifies the sampling weight of each vertex $j$ from $x_{ij}$ to $\nicefrac{x_{ij}}{q(y^{i-1}_j)}$ (for the counterpart of $q$ therein with a different functional form).
 	This adjustment cancels out many terms after the application of Jensen's inequality and leads to a different bound for $\Pr \big[ S \subset U^i \big]$.
 	
 	In any case, to prove the lemma for $i$, it remains to show that:
 	\[
 	\frac{\sum_{j \notin S} x_{ij} q(y^{i-1}_j)}{\sum_{j \in S} x_{ij} + \sum_{j \notin S} x_{ij} q(y^{i-1}_j)} \le \prod_{k \in S} \frac{q(y^{i-1}_k + x_{ik})}{q(y^{i-1}_k)}
 	~.
 	\]
 	
 	If an offline vertex $j$ has $x_{ij} = 0$, it plays no role in the above inequality:
 	on the left it contributes $0$ to both the numerator and the denominator, while on the right it corresponds to a multiplicative factor of $1$.
 	We will therefore exclude such vertices from the ranges of sums and products hereafter.
 	
 	We next crucially use the fact that $x_{ij}$'s are chosen by the unbounded \balance algorithm.
 	If $x_{ij} > 0$, the definition of unbounded \balance implies that $y^{i-1}_j + x_{ij} = \bar{y}$ for a common threshold $\bar{y} > 0$.
 	Further let $z = \sum_{j \notin S} x_{ij}$.
 	We can rewrite and bound the left-hand-side as:
 	\begin{align*}
 		\frac{\sum_{j \notin S} x_{ij} q(\bar{y} - x_{ij})}{\sum_{j \in S} x_{ij} + \sum_{j \notin S} x_{ij} q(\bar{y} - x_{ij})}
 		&
 		\le
 		\frac{\sum_{j \notin S} x_{ij} q\big((\bar{y} - z)^+\big)}{\sum_{j \in S} x_{ij} + \sum_{j \notin S} x_{ij} q\big((\bar{y} - z)^+\big)} \\
 		&
 		=
 		\frac{z q\big((\bar{y}-z)^+\big)}{1 - z + z q\big((\bar{y}-z)^+\big)}
 		~.
 	\end{align*}
 	
 	By the concavity of $\ln q(y)$, the right-hand-side can be rewritten and bounded as:
 	\begin{align*}
 		\prod_{k \in S} \frac{q(\bar{y})}{q(\bar{y}-x_{ik})}
 		&
 		= \exp \bigg( \sum_{k \in S} \big( \ln q(\bar{y}) - \ln q(\bar{y}-x_{ik}) \big) \bigg) \\
 		&
 		\ge \exp \bigg( \sum_{k \in S} x_{ik} \cdot \frac{q'(\bar{y})}{q(\bar{y})} \bigg) %
 		= \exp \bigg( (1-z) \cdot \frac{q'(\bar{y})}{q(\bar{y})} \bigg)
 		~.
 	\end{align*}
 	
 	Hence, it suffices to prove that:
 	\[
 	\frac{z q\big((\bar{y}-z)^+\big)}{1 - z + z q\big((\bar{y}-z)^+\big)} \le \exp \bigg( (1-z) \cdot \frac{q'(\bar{y})}{q(\bar{y})} \bigg)
 	~,
 	\]
 	which follows by Lemma~\ref{lem:balanceswor-condition}.
 \end{proof}

%% file: discussion.tex
\section{Discussion}
\label{sec:discussion}

We conclude the paper with a brief discussion on (1) our intuition of why rounding by sampling without replacement may be better than online correlated selection for online matching problems, and (2) the relevance of deterministic online algorithms.

Online correlated selection, by design, leads to a stronger notion of competitiveness:
for every offline vertex, the probability that it is matched at the end is at least $\Gamma$ times how much it is fractionally matched.
Sampling without replacement is only $1-\frac{1}{e}$ competitive w.r.t.\ this stronger notion.
To see this, consider the following example from \citet{Gao:FOCS:2021}.
Fix an offline vertex $j$.
Suppose that in each round $i = 1, 2, \dots, n$, vertex $j$ is paired with a fresh online vertex $j_i$ such that the fractional matching matches $j$ by $\frac{1}{n}$ and $j_i$ by $1-\frac{1}{n}$.
Hence, vertex $j$ is matched at the end with probability only $1 - (1-\frac{1}{n})^n \approx 1-\frac{1}{e}$.
Online correlated selection algorithms adjust the weights to slighter favor vertex $j$ in later rounds if it is still unmatched, in order to be competitive for every vertex.
We observe that, however, the overall matching is in fact extremely good in this example.
Since we have a fresh offline vertex in every round, every online vertex gets matched with certainty.
Therefore, we believe that the adjustments to sampling weights may be redundant;
we just need an appropriate analysis framework that considers the overall matching quality like the potential analysis introduced in this paper, instead of considering each offline vertex separately.

Regarding deterministic versus randomized online algorithms, apparently different readers may have different opinions on the significance of derandomizing online algorithms.
On the one hand, there is a long history of research on derandomizing algorithms for different problems in theoretical computer science.
On the other hand, due to the practical relevance of online matching algorithms, we argue that deterministic algorithms are more desirable for their explainability.
Every decision made by \reggreedy is driven by the estimated costs of matching different offline vertices (i.e., the regularization term).
If someone believes that \reggreedy has made a mistake, she/he can inspect the design of the regularization term and make changes accordingly.
By contrast, decisions made by randomized online algorithms can be ``merely random coin flips''.

We hope that the results and techniques from this paper could serve as building blocks for further research on sampling without replacement as an online rounding strategy, and on the design and analysis of deterministic online algorithms for online problems under stochastic models.

%% file: app-experiment.tex
\section{Experiments}
\label{sec:experiment}

This section compares the empirical performance of \swor and \balanceswor with existing algorithms in previous theoretical results, on real-world datasets.

The graphs from these real-world datasets are not bipartite.
We follow an existing practice in the literature (e.g., \citet*{Borodin:JEA:2020}) and transform them into bipartite graphs by the duplicating method.
In other words, let there be two copies of each vertex in the dataset, one on each side of the bipartite graph.
Put an edge between two vertices from the two sides if there is an edge between the corresponding vertices in the dataset.
We treat the resulting bipartite graphs as type graphs of Online Stochastic Matching in Experiments 1 and 2, and as input graphs of Online Bipartite Matching in Experiment 3.

Below we list the algorithms to be evaluated in our experiments, classified into three categories.
For readers' reference, we also include the best-known competitive ratio in parentheses, rounded down to the third digit after the decimal point.
\begin{enumerate}
	\item Algorithms for Online Bipartite Matching:
	\begin{enumerate}[itemsep=0pt, topsep=0pt, parsep=0pt]
		\item \ranking ($0.632$), due to \citet*{Karp:STOC:1990}.\label{item:ranking}
		\item \mindegree ($0.5$), folklore, c.f., \citet*{Borodin:JEA:2020}.
		\item \balanceocs ($0.593$), due to \citet{Gao:FOCS:2021}.
		\item \balanceswor ($0.513$), from Section \ref{sec:balanceSWOR}.\label{item:balanceswor}
	\end{enumerate}
	
	\item Algorithms for Online Stochastic Matching with integral arrival rates (i.e., $\lambda_i = 1$ for all $i$):
	\begin{enumerate}[itemsep=0pt, topsep=0pt, parsep=0pt]
		\item \feldman ($0.670$), due to \citet*{FeldmanMMM:FOCS:2009}.
		\item \bahmani ($0.699$), due to \citet*{Bahmani:ESA:2010}.
		\item \haeupler ($0.703$), due to \citet*{Haeupler:WINE:2011}.
		\item \jailletlu ($0.729$), due to \citet*{Jaillet:MOR:2014} (Section 3).
		\item \brubach ($0.729$), due to \citet*{Brubach:ESA:2016}.
	\end{enumerate}
	
	\item Algorithms for Online Stochastic Matching:%
	\footnote{\citet{TangWW:STOC:2022}'s algorithm (0.704) is not included here since it requires the values of $\Theta(|I|^2)$ windowed estimators in one run, which cannot be efficiently computed under the experiment setting.}
	\begin{enumerate}[itemsep=0pt, topsep=0pt, parsep=0pt]
		\item \manshadi ($0.702$), due to \citet*{ManshadiOS:MOR:2012}; the same algorithm is $0.705$-competitive assuming integral arrival rates.
		\item \jlnonint ($0.706$), due to \citet*{Jaillet:MOR:2014} (Section 5).
		\item \correlated ($0.711$), due to \citet*{HuangS:STOC:2021}. \label{item:correlated-sampling}
		\item \tophalf ($0.706$), due to \citet*{HuangSY:STOC:2022} (Section 3).
		\item \poissonocs ($0.716$), due to \citet*{HuangSY:STOC:2022} (Section 4).
		\item \swor ($0.707$), from Section \ref{sec:stochasticSWOR}.
		\item \reggreedy ($0.707$), from Section \ref{sec:stochasticSWOR}. \label{item:regularized-greedy}
	\end{enumerate}
\end{enumerate}

The original definitions of Algorithms \ref{item:correlated-sampling}--\ref{item:regularized-greedy} require receiving as input the solution $(x_{ij})_{(ij) \in E}$ of the Natural LP (or an even more expressive LP such as the second-level LP in the Poisson LP hierarchy \cite{HuangSY:STOC:2022}).
The Natural LP has exponentially many constraints. 
Solving it in polynomial time using the ellipsoid method as suggested by \citet{HuangS:STOC:2021} is impractical for the real-world graphs that we consider.
Hence, in Experiment 1 we use Monte Carlo simulation (10000 times) to compute $x_{ij}$ as the fraction of simulations whose optimal matching matches offline vertex $j$ to an online vertex with type $i$.
Experiment 2 will provide supporting evidence of this alternative approach, by considering smaller graphs and comparing these algorithms' performance when we use the Natural LP solution with their performance when we use the Monte Carlo method.

\paragraph{Technical Details.} Our experiments are performed on a desktop computer with Intel Core i7-8700 clocked at 3.2GHz and 16GB RAM. The computer is running Windows 10 64-bit Pro edition. All algorithms are implemented in C++\footnote{\url{https://github.com/online-stochastic-matching-library/online-stochastic-matching-library}.} and compiled with MinGW-w64 10.0.0 + GCC 12.2.0\footnote{Available at \url{https://winlibs.com/}.}, with GNU Linear Programming Kit (GLPK)\footnote{GLPK for MinGW-w64 available at \url{https://packages.msys2.org/base/mingw-w64-glpk}.} for linear program computation, under the following compiler options:
\begin{lstlisting}[xleftmargin=\parindent]
	g++.exe stochastic_matching_library.cpp -static -O2 -lglpk
	-std=c++11 -Wall -o stochastic_matching_library.exe
\end{lstlisting}

\paragraph{Experiment 1.}
Our main experiment evaluates the performance of all algorithms listed above in Online Stochastic Matching, using the same $6$ datasets considered by \citet*{Borodin:JEA:2020}, and assuming integral arrival rates $\lambda_i = 1$ for all online vertex type $i$.
The raw graph data are extracted from the Network Data Repository \cite{Ryan:AAAI:2015} and the information is listed in Table \ref{tab:reallifegraphstatistics}.

\begin{table}[htbp]
	\centering
	\caption{Real-World Graph Statistics}
	\label{tab:reallifegraphstatistics}
	\begin{tabular}{llllll}
		\toprule
		Datasets & Application & Nodes & Edges & Max Degree & Avg. Degree   \\
		\midrule
		SOCFB-CALTECH36 & social network & 769 & 17k & 248 & 43  \\
		SOCFB-REED98 & social network & 962 & 19k & 313 & 39 \\
		BIO-CE-GN &  biological network  & 2k  &  54k  &  242  & 48 \\
		BIO-CE-PG & biological network  & 2k  & 48k   & 913  & 51 \\
		ECON-MBEAFLW & economic network & 496 & 50k & 681& 201 \\
		ECON-BEAUSE & economic network & 507 & 44k & 766 & 175 \\
		\bottomrule
	\end{tabular}
\end{table}

This experiment generates $10,000$ random bipartite graphs from each type graph, and then for each algorithm evaluates the ratio of the average size of the algorithm's matching to the average size of the offline optimal matching.
The results are in Table \ref{tab:experiment-1-result}, and are accurate up to $\pm 0.001$ for a $95\%$ confidence level.

\begin{table}[htbp]
        \centering
		\caption{Results of Experiment 1 -- Online Stochastic Matching.}
		\label{tab:experiment-1-result}
		\begin{tabular}{lcccccc}
			\toprule
			Algorithm &Caltech36 &Reed98 &CE-GN &CE-PG &beause &mbeaflw\\
			\midrule
			\textbf{Regularized Greedy} & 0.928 & 0.929 & 0.984 & 0.990 & 0.962 & 0.966\\
			\textbf{Stochastic SWOR} & 0.929 & 0.927 & 0.958 & 0.962 & 0.959 & 0.975\\
			Poisson OCS & 0.929 & 0.926 & 0.957 & 0.960 & 0.958 & 0.974\\
			Min Degree & 0.879 & 0.873 & 0.948 & 0.955 & 0.952 & 0.975\\
			\textbf{Balance SWOR} & 0.874 & 0.873 & 0.943 & 0.950 & 0.943 & 0.971\\
			Balance OCS & 0.871 & 0.870 & 0.942 & 0.949 & 0.942 & 0.970\\
			Ranking & 0.859 & 0.859 & 0.934 & 0.944 & 0.936 & 0.966\\
			Heaupler et al & 0.829 & 0.829 & 0.837 & 0.860 & 0.809 & 0.797\\
			Top Half Sampling & 0.800 & 0.799 & 0.817 & 0.838 & 0.810 & 0.811\\
			Correlated Sampling & 0.786 & 0.785 & 0.809 & 0.836 & 0.780 & 0.771\\
			Manshadi et al & 0.786 & 0.785 & 0.809 & 0.836 & 0.780 & 0.771\\
			Brubach et al & 0.795 & 0.795 & 0.797 & 0.822 & 0.775 & 0.764\\
			Jaillet Lu (Integral) & 0.795 & 0.796 & 0.798 & 0.822 & 0.768 & 0.758\\
			Jaillet Lu & 0.773 & 0.773 & 0.789 & 0.816 & 0.763 & 0.754\\
			Bahmani Kapralov & 0.772 & 0.771 & 0.773 & 0.803 & 0.746 & 0.737\\
			Feldman et al & 0.766 & 0.766 & 0.771 & 0.801 & 0.742 & 0.737 \\
			\bottomrule
		\end{tabular}
\end{table}

\reggreedy and \swor are the best performing algorithms, and \poissonocs is a close third place.
Their results are noticeably better than the rest. 
Among algorithms that do not require stochastic information, \mindegree and \balanceswor are the top two algorithms whose results are slightly better than the results of \balanceocs and \ranking.
Our theoretical result (Section~\ref{sec:balanceSWOR}) shows that \balanceswor has nontrivial theoretical guarantees.
By contrast, \mindegree's worst-case competitive ratio is the trivial $0.5$.
It is an interesting direction to prove non-trivial competitive ratio for the \mindegree algorithm on some notion of ``natural instances''.

\paragraph{Experiment 2.}

This experiment consider several small real-world graphs for which we can solve the Natural LP using the ellipsoid method.
The raw graph data of these small graphs are also extracted from the Network Data Repository \cite{Ryan:AAAI:2015}, and their information is presented in Table \ref{tab:smallreallifegraphstatistics}.

\begin{table}[htbp]
	\centering
	\caption{Small Real-World Graph Statistics}
	\label{tab:smallreallifegraphstatistics}
	\begin{tabular}{l@{\hskip 24pt}l@{\hskip 6pt}l@{\hskip 6pt}l@{\hskip 6pt}l@{\hskip 6pt}l@{\hskip 6pt}l}
		\toprule
		Datasets & Application & Nodes & Edges & Max Degree & Avg. Degree   \\
		\midrule
		SOC-FIRM-HI-TECH & social network & 36 & 147 & 28 & 8  \\
		SOC-PHYSICIANS & social network & 241 & 1098 & 34 & 9 \\
		GENT-113 &  miscellaneous network  & 113  &  655  &  28  & 12 \\
		LP-BLEND & miscellaneous  network  & 114 & 522 & 31  & 9 \\
		\bottomrule
	\end{tabular}
\end{table}

We compare the performance of algorithms \ref{item:correlated-sampling}--\ref{item:regularized-greedy} when we use the Natural LP solution as the reference fractional matching $(x_{ij})_{(ij) \in E}$, with their performance when we use the Monte Carlo method instead.
This experiment generates $10,000$ random graphs for each type graph, and then evaluates the ratio of the average size of the algorithm's matching to the average size of the offline optimal matching.
The results are presented in Table \ref{tab:experiment-2-result}, and are accurate up to $\pm 0.001$ for a $95\%$ confidence level.

\begin{table}[htbp]
	\begin{center}
		\caption{Results of Experiment 2 -- Monte Carlo simulation (red color on the left) vs. Natural LP (blue color on the right).}
		\label{tab:experiment-2-result}
		\begin{tabular}{lcccc}
			\toprule
			Algorithm & hi-tech & physicians & gent113 & lp-blend\\
			\midrule
			\textbf{Regularized Greedy} & \textcolor{red}{0.955}/\textcolor{blue}{0.938} & \textcolor{red}{0.947}/\textcolor{blue}{0.937} & \textcolor{red}{0.957}/\textcolor{blue}{0.956} & \textcolor{red}{0.966}/\textcolor{blue}{0.976} \\
			\textbf{Stochastic SWOR} & \textcolor{red}{0.929}/\textcolor{blue}{0.930} & \textcolor{red}{0.927}/\textcolor{blue}{0.921} & \textcolor{red}{0.927}/\textcolor{blue}{0.920} & \textcolor{red}{0.948}/\textcolor{blue}{0.941} \\
			Poisson OCS & \textcolor{red}{0.928}/\textcolor{blue}{0.929} & \textcolor{red}{0.925}/\textcolor{blue}{0.920} & \textcolor{red}{0.926}/\textcolor{blue}{0.918} & \textcolor{red}{0.946}/\textcolor{blue}{0.937} \\
			Top Half Sampling & \textcolor{red}{0.836}/\textcolor{blue}{0.827} & \textcolor{red}{0.831}/\textcolor{blue}{0.804} & \textcolor{red}{0.836}/\textcolor{blue}{0.825} & \textcolor{red}{0.836}/\textcolor{blue}{0.867} \\
			Correlated Sampling & \textcolor{red}{0.831}/\textcolor{blue}{0.823} & \textcolor{red}{0.818}/\textcolor{blue}{0.798} & \textcolor{red}{0.833}/\textcolor{blue}{0.812} & \textcolor{red}{0.837}/\textcolor{blue}{0.845} \\
			\bottomrule
		\end{tabular}
	\end{center}
\end{table}

Generally speaking, these algorithms show similar performance with both methods.
We also observe that the Monte Carlo method yields slightly better overall results.
A possible explanation is that the exact probabilities of matching each offline vertex $j$ to each online type $i$ in the optimal matching satisfy constraints beyond the ones in the Natural LP, e.g., those from higher-level LPs in the Poisson LP Hierarchy~\cite{HuangSY:STOC:2022}.
The Monte Carlo method gives unbiased estimators of these exact probabilities and hence partially enjoys the power of these additional constraints.
It is another interesting research direction to improve our theoretical analysis by using these additional constraints.
Regardless of the method used, \reggreedy, \swor, and \poissonocs consistently outperform the other algorithms, and their relative order matches the results in Experiment 1 (Table \ref{tab:experiment-1-result}).

\paragraph{Experiment 3.}
Algorithms \ref{item:ranking}--\ref{item:balanceswor} are designed for Online Bipartite Matching, which considers an arbitrary arrival order.
Therefore, our last experiment considers the $6$ real-world graphs shown in Table \ref{tab:reallifegraphstatistics}, and evaluates the algorithms' worst performance w.r.t.\ each real-world graph out of $1000$ randomly generated arrival orders.
For each algorithm and each arrival order, we execute the algorithm $100$ times and consider its average performance.
The results are in Table \ref{tab:experiment-3-result}.

\begin{table}[htbp]
	\begin{center}
		\caption{Results of Experiment 3 -- Online Bipartite Matching.}
		\label{tab:experiment-3-result}
		\begin{tabular}{lcccccc}
			\toprule
			Algorithm &Caltech36 &Reed98 &CE-GN &CE-PG &beause &mbeaflw\\
			\midrule
			\textbf{Balance SWOR} & 0.840 & 0.830 & 0.924 & 0.924 & 0.919 & 0.956 \\
			Balance OCS & 0.835 & 0.828 & 0.922 & 0.923 & 0.918 & 0.956 \\
			Min Degree & 0.835 & 0.831 & 0.921 & 0.923 & 0.915 & 0.951 \\
			Ranking & 0.824 & 0.818 & 0.916 & 0.920 & 0.914 & 0.953 \\
			\bottomrule
		\end{tabular}
	\end{center}
\end{table}

\balanceswor performs the best in all $6$ graphs; \balanceocs is a close second place.
Recall that the only difference between \balanceswor and \balanceocs is that the latter complicates the matching probabilities of different unmatched offline vertices, in order to fit the existing analysis framework~\cite{Gao:FOCS:2021}.
The results of Experiment 3 suggest that in the context of online matching, such complications are likely a byproduct of the current analytical tools rather than a useful algorithmic technique.
The theoretical analysis in Section~\ref{sec:balanceSWOR} takes the first step toward understanding the simpler and empirically better \balanceswor algorithm.

%% file: app-swor.tex
\section{Missing Proofs from Section~\ref{sec:stochasticSWOR}}
\label{app:stochasticSWOR}

\subsection{Proof of Lemma \ref{lem:SWOR-reduc-feasibility-mass-distri}}
\label{subsec:app-swor-feasibility-proof}

\begin{proof}
	It suffices to prove the feasibility of the following LP problem:
	
	\begin{lemma}
		\label{lem:SWOR-reduc-feasibility-mass-distri-with-fewer-new-types}
		There exist nonnegative $\big\{\lambda_{i_j}, \rho_{{i_j}k}\big\}_{j, k\in S}$
		such that
		$\sum_{j\in S}\lambda_{i_j}=\lambda_i$ and
		for each $j\in S$:
		\begin{equation*}
			\rho_{i_j} = \rho_i~, \quad
			\rho_{{i_j}j}=1-\theta~, \quad
			\sum_{k\in S} \lambda_{i_k}\rho_{{i_k}j} = \lambda_i\rho_{ij}~,
		\end{equation*}
		and:
		\begin{equation}
			\label{eqn:SWOR-app-lower-bound-lambda_i_j}
			\lambda_{i_j}\ge
			\lambda_i\max\left\{\frac{p(\rho_i)-p(\rho_i-\rho_{ij})}{p(\rho_i)-p(\rho_i-1+\theta)},
			\frac{\rho_{ij}-\rho_i+1-\theta}{2-2\theta-\rho_i}\right\}~.
		\end{equation}
	\end{lemma}

	Then, by fixing a feasible solution $\big\{\lambda_{i_j}, \rho_{{i_j}k}\big\}_{j, k\in S}$ to the above LP and defining for any $j,k\in S, j\ne k$:
	\[\lambda_{(j,k)}=\frac{\lambda_{i_j} \rho_{{i_j}k}}{\rho_i-1+\theta}~,
	\]
	the following constraints on 
	$\big\{\lambda_{(j,k)}\big\}_{j,k\in S\colon j\ne k}$ 
	hold for any $j\in S$:
	\[
	\sum_{k\in S\colon k\ne j} \lambda_{(j,k)} (1-\theta)+
	\sum_{k\in S\colon k\ne j} \lambda_{(k,j)} (\rho_i-1+\theta)
	= \sum_{k\in S} \lambda_{i_k}\rho_{{i_k}j}
	= \lambda_i \rho_{ij}~,
	\]
	and
	\[
	\sum_{k \in S \colon k \ne j} \lambda_{(j,k)} \big( 1 - p(\rho_i - 1 + \theta) \big) 
	= \lambda_{i_j}\big(1-p(\rho_i-1+\theta)\big)
	\ge \lambda_i \big( p(\rho_i) - p(\rho_i - \rho_{ij}) \big)~.
	\]
\end{proof}

\subsection{Proof of Lemma \ref{lem:SWOR-reduc-feasibility-mass-distri-with-fewer-new-types}}
\begin{proof}
	We will show that the sum of the right-hand-side of Eqn.~\eqref{eqn:SWOR-app-lower-bound-lambda_i_j} over all $j\in S$ is at most $\lambda_i$. Then, there exists nonnegative $\big\{\lambda_{i_j}\}_{j\in S}$ that satisfies $\sum_{j\in S}\lambda_{i_j}=\lambda_i$ and Eqn.~\eqref{eqn:SWOR-app-lower-bound-lambda_i_j}.
	
	Observe that the right-hand side of Eqn.~\eqref{eqn:SWOR-app-lower-bound-lambda_i_j} equals: 
	\[\lambda_i\max\left\{
	\max\Big\{
	\frac{\rho_{ij}-\rho_i+\theta}{1-\rho_i},0
	\Big\},
	\max\Big\{
	\frac{\rho_{ij}-\rho_i+1-\theta}{2-2\theta-\rho_i},0
	\Big\}\right\}~.\]
	For any fixed value of $\rho_{i}$, the two functions within the outer max function are both of the form $\max\{m\rho_{ij}+c,0\}$, where $m>0$ and $c<0$. Moreover, they agree at $\rho_{ij}=1-\theta$. Hence, one of the two functions is always larger than or equal to the other one for any $\rho_{ij}\in [0,1-\theta]$.
	Further, since both functions are convex in $\rho_{ij}$, the sums $\sum_{j\in S}
	\lambda_i\max\big\{
	\frac{\rho_{ij}-\rho_i+\theta}
	{1-\rho_i},0
	\big\}$ 
	and
	$\lambda_i\sum_{j\in S}
	\max\big\{
	\frac{\rho_{ij}-\rho_i+1-\theta}{2-2\theta-\rho_i},0
	\big\}$
	are maximized when there are $j_1,j_2$ such that $\rho_{ij_1}=1-\theta$ and $\rho_{ij_2}=\rho_{i}-1+\theta$, both with maximum value $\lambda_i$.
	Therefore, the result follows.
	
	Fix a $\big\{\lambda_{i_j}\}_{j\in S}$ that satisfies $\sum_{j\in S}\lambda_{i_j}=\lambda_i$. 
	Recall that for any $j\in S$,
	$\rho_{i_j}=\rho_i$ and $\rho_{{i_j}j}=1-\theta$ are fixed.
	We next determine the values of $\rho_{{i_j}k}$ for any $j,k\in S, j\ne k$ by solving the following max flow problem:
	
	Construct a network $G=(V,E)$, with $s,t\in V$ being the source and the sink respectively, according to the following steps:
	\begin{itemize}
		\item For every $j\in S$, add node $j$ to $V$ and then edge $(s,j)$ to $E$ with capacity $\lambda_i\rho_{ij}-\lambda_{i_j}(1-\theta)$.
		\item For any $j\in S$, add node $i_j$ to $V$ and then edge $(i_j,t)$ to $E$ with capacity $\lambda_{i_j}(\rho_i-1+\theta)$.
		\item For any $j,k\in S$ with $j\ne k$, add edge $(j,i_k)$ to $E$ with infinite capacity.
	\end{itemize}
	
	We shall prove that the max flow has value $\lambda_i (\rho_i-1+\theta)$.
	For simplicity, we define $T=\{i_j\colon j\in S\}$.
	Further, let $(U,\bar{U})$ be a min cut of this network. Note that it must satisfy exactly one of the following conditions, otherwise the cut would have infinite capacity:
	\begin{itemize}
		\item $U=\{s\}$
		\item $\bar{U}=\{t\}$
		\item $U=\{s,j\}\cup T\setminus \{i_j\}$ for some $j\in S$
		\item $\bar{U}=\{t,i_j\}\cup S\setminus \{j\}$ for some $j\in S$
	\end{itemize}
	
	For $U=\{s\}$, the capacity of the cut is $\sum_{j\in S} \lambda_i\rho_{ij}-\lambda_{i_j}(1-\theta)
	=\lambda_i(\rho_i-1+\theta)$.
	
	For $\bar{U}=\{t\}$, the capacity of the cut is
	$\sum_{j\in S}\lambda_{i_j}(\rho_i-1+\theta)
	=\lambda_i(\rho_i-1+\theta)$.
	
	Now, fix an offline vertex $j$.
	Notice that the cut set for the case with $U=\{s,j\}\cup T\setminus \{i_j\}$ and the case with $\bar{U}=\{t,i_j\}\cup S\setminus \{j\}$ are both $\{(s,k), (i_k,t)\colon k\in S\setminus\{j\}\}$.
	Therefore, the capacities of both cuts are
	\begin{equation*}
		\begin{aligned}
			&\hphantom{=}\sum_{k\in S\colon k\ne j} \lambda_i\rho_{ik}-\lambda_{i_k}(1-\theta)
			+\sum_{k\in S\colon k\ne j} \lambda_{i_k}(\rho_i-1+\theta) \\
			&=	\lambda_i(\rho_i-\rho_{ij})-(\lambda_i-\lambda_{i_j})(1-\theta)+(\lambda_i-\lambda_{i_j})(\rho_i-1+\theta)\\
			&=\lambda_i\big(2(\rho_i-1+\theta)-\rho_{ij}\big)+\lambda_{i_j}\big(2-2\theta-\rho_i\big)~,
		\end{aligned}
	\end{equation*}
	which is at least $\lambda_i(\rho_i-1+\theta)$ by Eqn.~\eqref{eqn:SWOR-app-lower-bound-lambda_i_j}.
	As a result, the min cut capacity of network $G$, and thus the max flow value, is exactly $\lambda_i(\rho_i-1+\theta)$.
	
	Let $f\colon E\rightarrow \mathbb{R}$ be a max flow of the network $G$. The above argument implies that for any $j\in S$, the flow $f$ from $s$ to $j$ and the flow from $i_j$ to $t$ are both at full capacity. By defining $\rho_{{i_j}k}:=\frac{f(k,i_j)}{\lambda_{i_j}}$ for any $j,k\in S$ with $j\ne k$,
	it is easy to check that
	$\rho_{i_j}:=\sum_{k\in S} \rho_{{i_j}k} = \rho_i$ and 
	$\sum_{k\in S} \lambda_{i_k}\rho_{{i_k}j} = \lambda_i\rho_{ij}$ for any $j\in S$.	
\end{proof}

%% file: app-balance-swor.tex
\section{Missing Proofs from Section~\ref{sec:balanceSWOR}}
\label{app:balanceSWOR}

 \subsection{Proof of Lemma \ref{lem:ocs-to-ratio}}
\label{sec:proof-ocs-to-ratio}

\begin{proof}
	The expected size of matching is at least:
	\[
	\sum_{j \in J} (1 - q(y_j)).
	\]
	
	Similar to \citet{Gao:FOCS:2021}, we distribute the contribution of each online vertex to $i$ and all its offline neighbors, denoted by $\beta_i$ and $\alpha_j$ for each $i \in I, j \in J$.
	According to the algorithm, when $i \in I$ arrives, let threshold $\bar{y} > 0$ such that $\sum_{(i, j) \in E} (\bar{y} - y_j^{i-1})^{+}=1$. The algorithm allocates $x_{ij} = (\bar{y} - y_j^{i-1})^{+}$ for every neighbor $j$. Further define $\Gamma$ and function $a(y), b(y)$ as follows:
	\begin{align*}
		\Gamma = & \int_0^{\infty} e^{-z} (1 - q(y)) \mathrm{d}z ; & \qquad &\\
		b(y) = & -e^y \int_y^{\infty} q'(z) e^{-z} \mathrm{d}z & \forall y \geq 0;\\
		a(y) = & -q'(y)-b(y) & \forall y \geq 0.
	\end{align*}
	
	For each offline neighbor $j$, increase $\alpha_j$ by:
	\[
	\int_{y_j^{i-1}}^{y_j^i} a(z) \mathrm{d} z,
	\]
	and set $\beta_i$ to:
	\[
	\sum_{(i, j) \in E}\int_{y_j^{i-1}}^{y_j^i} b(z) \mathrm{d} z.
	\]
	
	$b(y)$ is non-increasing because:
	\begin{align*}
		b'(y) = & ~ - e^y \int_y^{\infty} q'(z)e^{-z} \mathrm{d}z + q'(y) \\
		= & ~ e^y \int_y^{\infty} (q'(y) - q'(z)) e^{-z} \mathrm{d}z\\
		\leq & ~ 0, \tag{Convexity of $q$}
	\end{align*}
	hence $\beta_i \geq b(\bar{y}) \cdot \sum_{(i, j) \in E} (\bar{y} - y_j^{i-1})^{+} = b(\bar{y})$, at least  $b(y_j)$ for any $(i, j) \in E$.
	
	Note that, the gain upon arrival of $i \in I$ is:
	\[
	\sum_{(i, j) \in E} \int_{y_j^{i-1}}^{y_j^i} \big( a(z) + b(z) \big) \mathrm{d} z = \sum_{(i, j) \in E} \int_{y_j^{i-1}}^{y_j^i} -q'(z) \mathrm{d} z = \sum_{(i, j) \in E} \big(  (1 - q(y_j^i)) - (1 - q(y_j^{i-1})) \big).
	\]
	
	Summing up the total distribution of all vertices, it is at most:
	\[
	\sum_{j \in J} (1 - q(y_j)),
	\]
	no more than the expected size of matching.
 
	On the other hand, the distribution satisfies an approximate equilibrium condition in the sense that for any edge $(i, j)$, the total distribution of $i$ and $j$ is at least:
	\[
	\int_{0}^{y_j} a(z) \mathrm{d} z + b(y_j) = ~ b(0) + \int_{0}^{y_j} (a(z) + b'(z)) \mathrm{d} z \geq \Gamma, 
	\]
	since 
	\begin{align*}
		a(y) + b'(y) & = ~ a(y) - e^y \int_y^{\infty} q'(z) e^{-z} \mathrm{d}z + q'(y) \\
		& = ~ a(y) + b(y) + q'(y) \\
		& = ~ 0
	\end{align*}
	and
	\[
	b(0) = \int_0^{\infty} q'(z) e^{-z} \mathrm{d}z = \int_0^{\infty} e^{-z} (1 - q(z)) \mathrm{d}z = \Gamma
	\]
	hold from integration by parts.
	
	Therefore, for any optimal matching $M$, the competitive ratio of $\balanceswor$ is at least $\Gamma$ because:
	\begin{align*}
		\textsc{ALG} & \geq \sum_{j \in J} \alpha_j + \sum_{i \in I} \beta_i\\
		& \geq \sum_{(i, j) \in M} (\alpha_j + \beta_i)\\
		& \geq \Gamma \cdot |M|.
	\end{align*}
\end{proof}

\subsection{Proof of Lemma~\ref{lem:balanceswor-condition}}

\begin{proof}   
	If $y \le 1$, we have $q(y) = e^{-y}$ and $q'(y) = - e^{-y}$.
	The right-hand-side is therefore $e^{-(1-z)}$.
	Further, the left-hand-side is increasing in $q\big((y-z)^+\big)$ and hence is maximized when $y \le z$ and $q\big((y-z)^+\big) = 1$.
	It reduces to a basic inequality:
	\[
	z \le e^{-(1-z)}
	~,
	\]
	which holds for any real number $z$, and takes equality when $z = 1$.
	
	Next consider $y > 1$.
	We have $q(y) = e^{-e^{y-1}}$ and $q'(y) = - q(y) e^{y-1}$.
	Hence, the right-hand-side equals $\exp\big(-(1-z) e^{y-1} \big)$.
	Further note that $y > 1$ implies $y - z > 0$.
	The inequality reduces to:
	\[
	\frac{z q(y-z)}{1-z+zq(y-z)} \le \exp \big( -(1-z) e^{y-1} \big)
	\quad
	\Leftrightarrow
	\quad
	\exp \big( (1-z) e^{y-1} \big) \le 1 + \frac{1-z}{z q(y-z)}
	~.
	\]
	
	We next change variables by letting $y' = y-1$ and $z' = 1-z$. 
	We have $y' > 0$, $0 \le z' \le 1$, and $y-z = y'+z'$.
	The inequality becomes:
	\[
	e^{z'e^{y'}} \le 1 + \frac{z'}{1-z'} \frac{1}{q(y'+z')}
	~.
	\]
	
	There are two subcases depending on the value $y'+z'$.
	\paragraph{Subcase 1.}
	If $y' + z' \le 1$, we have $q(y'+z') = e^{-(y'+z')}$.
	The inequality reduces to:
	\[
	e^{z'e^{y'}} \le 1 + \frac{z'}{1-z'} e^{y'+z'}
	~.
	\]
	
	Since the left-hand side is convex in $e^{y'}$ and the right-hand side is linear, it suffices to consider the two extremes when $y' = 0$ and $y' = 1-z'$.
	When $y' = 0$, it becomes:
	\[
	e^{z'} \le 1 + \frac{z'}{1-z'} e^{z'}
	~,
	\]
	which holds for any $0 \le z' \le 1$, and takes equality at $z' = 0$.
	To prove it, we rearrange terms to convert it into:
	\[
	\frac{1-2z'}{1-z'} \le e^{-z'}
	~,
	\]
	which follows by $e^{-z'} \ge 1-z'$.
	
	When $y' = 1-z'$, it becomes:
	\[
	e^{z'e^{1-z'}} \le 1 + \frac{ez'}{1-z'}
	~,
	\]
	which holds for any $0 \le z' \le 1$, and takes equality at $z' = 0$.
	To prove this, we first take the logarithm of both sides to convert it into:
	\[
	z' e^{1-z'} \le \log \frac{1+(e-1)z'}{1-z'}
	~.
	\]
	The derivative of the left-hand-side is:
	\[
	e^{1-z'}(1-z') \le e(1-z')
	~.
	\]
	The derivative of the right-hand-side is:
	\[
	\frac{e-1}{1+(e-1)z'} + \frac{1}{1-z'} = \frac{e}{(1+(e-1)z')(1-z')} \ge e(1-z')
	~,
	\]
	where the inequality follows by $\frac{1}{(1-z')^2} \ge 1+2z' \ge 1+(e-1)z'$.
	Hence, it suffices to verify the inequality at $z' = 0$, where it holds with equality.

	\paragraph{Subcase 2.}
	Finally, consider the case when $y' + z' > 1$, where we have $q(y'+z') = e^{-e^{y'+z'-1}}$.
	The inequality reduces to:
	\[
	e^{z'e^{y'}} \le 1 + \frac{z'}{1-z'} e^{e^{y'+z'-1}}
	~.
	\]
	
	We now change variables again by letting $u = e^{z' e^{y'}}-1$.
	The inequality becomes:
	\[
	\frac{1-z'}{z'} u \le (1+u)^{\frac{e^{z'-1}}{z'}}
	~.
	\]
	
	Let $\alpha = z' e^{1-z'}$.
	Note that $0 \le \alpha \le 1$.
	By Young's inequality:
	\[
	1+u = (1-\alpha) \cdot \frac{1}{1-\alpha} + \alpha \cdot \frac{u}{\alpha} \ge \frac{u^\alpha}{(1-\alpha)^{1-\alpha} \alpha^\alpha}
	~.
	\]
	
	We therefore only need to prove that:
	\[
	\frac{1-z'}{z'} \le \frac{1}{(1-\alpha)^{\frac{1-\alpha}{\alpha}} \alpha}
	~.
	\]
	
	We first claim that:
	\[
	\frac{1-z'}{z'} \alpha \le e - 2\alpha
	~,
	\]
	that is:
	\[
	(1-z')e^{1-z'}\le e - 2z' e^{1-z'}
	~,
	\]
	which holds by:
	\[
	1+z'\le e^{z'}
	~.
	\]
	
	Therefore it suffices to prove for any $\alpha \in [0, 1]$:
	\[
	(e - 2 \alpha) (1-\alpha)^{\frac{1-\alpha}{\alpha}} \le 1
	~.
	\]
	or equivalently:
	\[
	(e - 2 \alpha)^\alpha (1-\alpha)^{1-\alpha} \le 1
	~.
	\]
	
	We know that:
	\begin{align*}
		(e - 2 \alpha)^\alpha (1-\alpha)^{1-\alpha} & = e^{\alpha} (1 - \frac{2}{e} \alpha)^\alpha (1-\alpha)^{1-\alpha} \\
		& \le e^{\alpha} \cdot \bigg(\alpha(1 - \frac{2}{e}\alpha) + (1-\alpha)(1-\alpha)\bigg)\\
		& = e^{\alpha} \cdot (1 - \alpha + \frac{e-2}{e} \alpha^2),
	\end{align*}
	where the second inequality is due to Young's inequality, therefore we only need to prove:
	\[
	1 - \alpha + \frac{e-2}{e} \alpha^2 \le e^{-\alpha}
	~,
	\]
	which holds for all $\alpha \in [0, 1]$.
\end{proof}

%% file: app-regularized-greedy-weighted.tex
\section{Potential Function Analysis for Edge-Weighted Matching}
\label{sec:extension}

This section shows the adaptability of potential function analysis in online edge-weighted stochastic matching with free disposal. We present how to apply this analysis framework to derandomize the 0.706-competitive \tophalf originally developed by \citet{HuangSY:STOC:2022} while preserving its competitive ratio.

The edge-weighted problem with free disposal generalizes the unweighted problem in the following sense. For each edge $e=(i,j)\in E$, let $w_{ij}>0$ be the weight of matching an online vertex of type $i$ to the offline vertex $j$. The goal is then to maximize the expected total weight of the matched edges. Under the assumption of free disposal, each offline vertex can be matched more than once, but only the weight of its heaviest edge contributes to the objective. Thus, if an offline vertex $j$ has been matched to an online vertex of type $i$, a rematch between vertex $j$ and an online vertex of type $i'$ with edge weight $w_{i'j} > w_{ij}$ leads to an objective gain of $w_{i'j}-w_{ij}$.
Effectively, the algorithm disposes edge $(i,j)$ for free.
The problem is vertex-weighted if $w_{ij}=w_j$ for any edge $(i,j)$, and is unweighted if $w_j=1$ for any offline vertex $j$. 
In vertex-weighted and unweighted matching, free disposal is irrelevant because all incident edges of an offline vertex share the same weight so the algorithm has no reason to rematch any offline vertex.

We first review the definition of \tophalf. Suppose that an online vertex of type $i$ arrives at time $0\leq t \leq 1$. For each offline neighbor $j\in J_i$, let $w_j(t)$ be the maximum edge weight matched to $j$ right before time $t$. Then the marginal weight of edge $(i,j)$ is given by:
\[w_{ij}(t)\overset{\text{def}}{=}(w_{ij}-w_j(t))^+~.\]

Next, we specify a total order $\succ_{i,t}$ over offline neighbors $J_i$ of $i$ in descending order of the marginal weights, breaking ties arbitrarily. In other words, if $j\succ_{i,t} j'$, then $w_{ij}(t)\geq w_{ij'}(t)$. We write $j\succeq_{i,t} j'$ if either $j\succ_{i,t} j'$ or $j=j'$.

Towards the end, we define the mapping $\sigma_{i,t}\colon [0,\lambda_i)\to J_i\cup\{\perp\}$, where $\perp$ is a dummy offline vertex with zero marginal weight and $x_{i\perp}=\lambda_i - \sum_{j\in J_i} x_{ij}\geq 0$. For any $\tau \in [0,\lambda_i)$, we define $\sigma_{i,t}(\tau)$ such that:
\[
\sum_{j\in J_i\colon j\succ_{i,t} \sigma_{i,t}(\tau)} x_{ij} \leq \tau
< \sum_{j\in J_i\colon j\succeq_{i,t} \sigma_{i,t}(\tau)} x_{ij}~.
\]

\tophalf samples $\tau$ uniformly from $\big[0,\frac{\lambda_i}{2}\big)$ and subsequently matches the online vertex to $j=\sigma_{i,t}(\tau)$. 
\footnote{
	Following the definition of $p$ and $\rho_i(t)$ in Section \ref{subsec:analysis-framework}, if we let $\theta=\frac12$, then $\lambda_i p(\rho_i(t))$ equals \tophalf's match rate of online vertex type $i$ at time $t$. Note that the choice of sampling $\tau$ from the top half of $[0,\lambda_i)$ (and thus the choice of $\theta=\frac12$) is not optimal. However, since optimizing the parameter offers only a marginal improvement and requires a more intricate proof, we maintain the choice of $\theta = \frac12$, as previously considered in \citet{HuangSY:STOC:2022}.
}
Since the offline neighbors are sorted in descending order of the marginal weights, this sampling doubles the probability of matching $(i,j)$ suggested by the Natural LP solution for the more valuable neighbor(s) $j$ of $i$, while the matching probability of $(i,j')$ for the less valuable neighbor(s) $j'$ may get truncated if the total doubled probability exceeds $1$.

\begin{tcolorbox}[beforeafter skip=10pt]
	\textbf{Top Half Sampling} \cite{HuangSY:STOC:2022}\\[1ex]
	\emph{Input at the beginning:}
	\begin{itemize}[itemsep=0pt, topsep=4pt]
		\item Type graph $G = (I, J, E)$;
		\item Arrival rates $(\lambda_i)_{i \in I}$;
		\item Solution of the Natural LP $(x_{ij})_{(i,j) \in E}$.
	\end{itemize}
	\smallskip
	\emph{When an online vertex of type $i \in I$ arrives at time $0 \le t \le 1$:}
	\begin{itemize}[itemsep=0pt, topsep=4pt]
		\item Sample $\tau$ uniformly from $\big[0,\frac{\lambda_i}{2}\big)$;
		\item Match the online vertex to $j=\sigma_{i,t}(\tau)$.
	\end{itemize}
\end{tcolorbox}

For any time $0\leq t\leq 1$, let $A(t)$ denote the total weight of of edges matched by the online algorithm right before time $t$.
Following the potential function analysis framework and the design of \reggreedy in Section \ref{sec:stochasticSWOR}, we shall design the deterministic algorithm along with a potential function $\pot(t)$ that satisfy Conditions \ref{con:start}-\ref{con:monotone} under the extended definition of $A(t)$.

First, we keep track of the unmatched portion of fractional matching $(x_{ij})_{(i,j)\in E}$ using the following notations:
\begin{equation*}
	x_{ij}(t) =
	\begin{cases}
		x_{ij} & \mbox{if $j$ is unmatched} \\
		0 & \mbox{otherwise}
	\end{cases}
	~,\quad
	x_j(t) = \sum_{i \in I_j} x_{ij}(t)
	~,\quad
	x(t) = \big\{ x_{ij}(t) \big\}_{(i,j) \in E}
	~.
\end{equation*}

We also decompose the marginal contribution of each edge towards the rate of change of the objective based on continuous weight levels.
For this purpose, we keep track of the unmatched portion of fractional matching for any weight level $w\geq 0$ using the following notations:
\begin{equation*}
	x_{ij}(t,w) =
	\begin{cases}
		x_{ij} & \mbox{if $w\leq w_{ij}(t)$} \\
		0 & \mbox{otherwise}
	\end{cases}
	~,\quad
	x_i(t,w) = \sum_{j\in J_i} x_{ij}(t,w)
	~.
\end{equation*}

Conditioned on any marginal weight $w(t)=\{w_{ij}(t)\}_{(i,j)\in E}$ by time $t$, we express the rate $X(t)$ at which the objective would increase if we match online vertices according to the solution of the Natural LP:
\begin{equation*}
	X(t)
	=\sum_{(i,j)\in E}
	w_{ij}(t) x_{ij}
	=\int_0^\infty \sum_{(i,j)\in E} x_{ij}(t,w) dw~.
\end{equation*}

Second, for any edge $(i,j)$, we define $y_{ij}(t)$ to be the match rate of $(i,j)$ at time $t$ if we match online vertices according to \tophalf, and $y_j(t) = \sum_{i\in I_j} y_{ij}(t)$. We keep track of the rate of increase of the objective for any weight level $w>0$ if we match online vertices according to \tophalf using the following notations:
\begin{equation*}
	y_{ij}(t,w) =
	\begin{cases}
		y_{ij}(t) & \mbox{if $w\leq w_{ij}(t)$} \\
		0 & \mbox{otherwise}
	\end{cases}
	~,\quad
	y_j(t,w) = \sum_{i\in I_j} y_{ij}(t,w)
	~.
\end{equation*}

Conditioned on any marginal weight $w(t)$ by time $t$, we express the rate $Y(t)$ at which the objective would increase if we match online vertices according to \tophalf:
\begin{equation*}
	Y(t)
	=\sum_{(i,j)\in E} w_{ij}(t) y_{ij}(t)
	=\int_0^\infty \sum_{(i,j)\in E} y_{ij}(t,w) dw~.
\end{equation*}

Our potential function is:
\begin{equation}
	\label{eqn:potential-edge-weighted}
	\pot(t)
	=\alpha(t)X(t)+\beta(t)Y(t)~,
\end{equation}
where $\alpha,\beta\colon [0,1] \to \mathbb{R}_+$ are defined in the same way as in Section \ref{subsec:analysis-framework} with parameter $\theta=\frac12$, i.e.:
\begin{align}
	\alpha(t)&=1-\frac{1}{e(1-\ln 2)}(2e)^t+\frac{1+\ln 2}{e^2(1-\ln 2)}e^{2t}~, \label{eqn:alpha_top_half}
	\\
	\beta(t)&=\frac{1}{2e(1-\ln 2)}(2e)^t-\frac{1}{e^2(1-\ln 2)}e^{2t}~. \label{eqn:beta_top_half}
\end{align}

Lemma \ref{lem:alpha-beta_top_half} lists all the necessary properties of $\alpha(t)$ and $\beta(t)$ for the competitive analysis of our algorithm. The proofs are basic calculus and hence omitted. In fact, Properties \ref{prop: start_0.5}, \ref{prop: end_0.5}, \ref{prop: derivative of alpha_0.5} and \ref{prop: derivative of beta_0.5} are exactly the same as in Lemma \ref{lem:alpha-beta}.

\begin{lemma}
	\label{lem:alpha-beta_top_half}
	The functions $\alpha(t),\beta(t)$ defined in Eqn.~\eqref{eqn:alpha_top_half} and \eqref{eqn:beta_top_half} satisfy that:
	\begin{enumerate}
		\item $\alpha(0)+\beta(0)=1-\dfrac{1}{1-\ln 2} \Big(\dfrac{1}{2e}-\dfrac{\ln 2}{e^2}\Big) > 0.706 ~.$
		\label{prop: start_0.5}
		\item $\alpha(1)=\beta(1)=0~.$
		\label{prop: end_0.5}
		\item $\alpha(t), \beta(t)\geq 0~$ for any $0\leq t\leq 1~.$
		\label{prop: pos_0.5}
		\item $\frac{d}{dt}\alpha(t)=-2(1+\ln 2)\beta(t)$ for any $0\le t \le 1~.$
		\label{prop: derivative of alpha_0.5}
		\item $\frac{d}{dt}\beta(t)=\alpha(t)+(3+\ln 2)\beta(t)-1$ for any $0\le t \le 1~.$
		\label{prop: derivative of beta_0.5}
	\end{enumerate}
\end{lemma}

For notational simplicity, we define auxiliary variables for match rates on each weight-level $w$, i.e., the values of $x_{ij}(t,w)$ divided by the arrival rate $\lambda_i$ of online vertex type $i$:
\begin{equation*}
	\rho_{ij}(t,w) = \frac{x_{ij}(t,w)}{\lambda_i}
	~,\quad
	\rho_i(t,w) = \sum_{j \in J_i} \rho_{ij}(t,w)
	~.
\end{equation*}

To complete the algorithm definition, we define a regularization term for any edge $(i,j)$ and by any time $0\leq t\leq 1$, depending on the marginal weights $w(t)$ at time $t$, and correspondingly how much $X(t)$ and $Y(t)$ would decrease should an online vertex of type $i$ be matched to offline vertex $j$ at time $t$:
\begin{equation}
	\label{eqn:regularization-top_half}
	\begin{multlined}
		r_{ij}(t)=\alpha(t)~ \underbrace{\sum_{i'\in I_j} \Big(w_{i'j}(t)-\big(w_{i'j}(t)-w_{ij}(t)\big)^+\Big) x_{i'j}}_{\mbox{\small decrease of $X(t)$}}
		\\[1ex]
		+
		\beta(t)~ \underbrace{\sum_{i'\in I_j} \int_{(w_{i'j}(t)-w_{ij}(t))^+}^{w_{i'j}(t)} \lambda_{i'}
			\Big( 
			p\big(\rho_{i'}(t,w)\big)
			-
			p\big(\rho_{i'}(t,w)-\rho_{i'j}(t,w)\big)
			\Big) dw}_{\mbox{\small decrease of $Y(t)$}}~,
	\end{multlined}
\end{equation}
where we define:
\[p(\rho) = \min \big\{ 2\rho, 1 \big\}~.\]

Note that in the unweighted setting, the regularization term simplifies from Eqn.~\eqref{eqn:regularization-top_half} to Eqn.~\eqref{eqn:regularization}.

We check that if $j$ is unmatched at time $t$, then for any $i'\in I_j$, $w_{i'j}(t)=1$ and $\rho_{i'j}(t,w)=\rho_{i'j}=\rho_{i'j}(t)$ for any $w\in [0,1]$, 
so the decrease of $X(t)$ equals $x_j = x_j(t)$,
while the decrease of $Y(t)$ becomes
$\sum_{i' \in I_j} \lambda_{i'} \Big( p \big(\rho_{i'}(t)\big) - p\big(\rho_{i'}(t) - \rho_{i'j}(t)\big)\Big)$ because the lower and the upper limits of each integral are $0$ and $1$ respectively.

Otherwise, $j$ is already matched, then every $w_{i'j}(t)$ and $\rho_{i'j}(t,w)$ are zero. Hence, both terms are zero, which again match the definition in Eqn.~\eqref{eqn:regularization}.

\begin{tcolorbox}[beforeafter skip=10pt]
	\textbf{\reggreedy (for Edge-Weighted Matching with Free Disposal)}\\[1ex]
	\emph{Input at the beginning:}
	\begin{itemize}[itemsep=0pt, topsep=4pt]
		\item Type graph $G = (I, J, E)$;
		\item Arrival rates $(\lambda_i)_{i \in I}$;
		\item Solution of the Natural LP $(x_{ij})_{(i,j) \in E}$.
	\end{itemize}
	\smallskip
	\emph{When an online vertex of type $i \in I$ arrives at time $0 \le t \le 1$:}
	\begin{itemize}[itemsep=0pt, topsep=4pt]
		\item Match it to an offline neighbor $j$
		that maximizes $w_{ij}(t)-r_{ij}(t)$, where $r_{ij}(t)$ is given in Eqn. \eqref{eqn:regularization-top_half}.
	\end{itemize}
\end{tcolorbox}

\begin{theorem}
	\label{thm:regularized-greedy-app}
	For any instance of edge-weighted online stochastic matching with free disposal, \reggreedy (for Edge-Weighted Matching with Free Disposal) is $0.706$-competitive.
\end{theorem}

The proof of this theorem is based on Lemma 5 from \citet{HuangSY:STOC:2022}. For the sake of completeness, we include its proof at the end of the section.

\begin{lemma}[c.f.\ \citet{HuangSY:STOC:2022}]
	\label{lem:HuangSY-STOC-2022-restate}
	For any $0\leq t\leq 1$:
	\begin{equation}
		\label{eq:restated-equation}
		\frac{d}{dt} \E\, \big[Y(t)\mid w(t)\big] \\
		\geq 2(1+\ln 2)X(t) - (3+\ln 2) Y(t)~.
	\end{equation}
\end{lemma}

\begin{proof}[Proof of Theorem \ref{thm:regularized-greedy-app}]
	Following the potential function analysis framework in Section \ref{subsec:analysis-framework}, it suffices to show that our potential function $\pot(t)$ given in Eqn.~\eqref{eqn:potential-edge-weighted} and the algorithm's weight of matching $A(t)$ satisfies Conditions \ref{con:start}, and \ref{con:end} and \ref{con:monotone} with $\Gamma = 0.706$.
	
	\paragraph{Condition \ref{con:start}.}
	Since \tophalf doubles the match rates of the most valuable half of neighbors of each online vertex type, we have:
	\[
	Y(0)
	=\sum_{(i,j)\in E}
	w_{ij}(0) y_{ij}(0)
	\geq \sum_{(i,j)\in E} w_{ij}(0) x_{ij}
	= X(0) = \opt
	~.
	\]
	Therefore, we get that:
	\begin{align*}
		\pot(0) 
		= \alpha(0) X(0) + \beta(0) Y(0)
		\ge \big(\alpha(0) + \beta(0)\big) \cdot \opt ~.
	\end{align*}
	Further, Property \ref{prop: start_0.5} of Lemma \ref{lem:alpha-beta_top_half} suggests that $\alpha(0) + \beta(0) > 0.706$, so the condition follows.

	\paragraph{Condition \ref{con:end}.}
	It follows by Property \ref{prop: end_0.5} of Lemma \ref{lem:alpha-beta_top_half}.
	
	\paragraph{Condition \ref{con:monotone}.}
	For any time $0\le t \le 1$, and conditioned on any marginal weights $w(t)$ at time $t$, we next show that the derivative of $\E\, [ \alg(t) + \pot(t) \mid w(t) ]$ is nonnegative.
	The derivative can be expressed as:
	\begin{equation}
		\label{eqn:SWOR-derivative-of-A-plus-Phi_top_half}
		\begin{multlined}
			\overbrace{\frac{d}{dt} \E\, \big[ \alg(t) \mid w(t) \big]
				+
				\alpha(t) \frac{d}{dt} \E\, \big[X(t)\mid w(t) \big]
				+
				\beta(t) \frac{d}{dt} \E\, \big[Y(t)\mid w(t)\big]}^{\mbox{\small maximized by Regularized Greedy}}
			\\[0.5ex]
			+\underbrace{\Big( \frac{d}{dt} \alpha(t) \Big) X(t)
				+\Big( \frac{d}{dt} \beta(t) \Big) Y(t)}_{\mbox{\small independent of algorithm's decisions}}~.
		\end{multlined}
	\end{equation}
	
	To formalize the remark above the first line of Eqn. \eqref{eqn:SWOR-derivative-of-A-plus-Phi_top_half}, let $\mu_{ij}^*(t)$ be the indicator of whether \reggreedy would match an online vertex $i$, arrived at time $t$, to offline vertex $j$. For concreteness, we assume that \reggreedy breaks ties by the lexicographical order of offline vertices, although the argument would work for an arbitrary tie-breaking rule. That is:
	\[\mu_{ij}^*(t)=
	\begin{cases}
		1 & \text{if }j\text{ is the first offline vertex that maximizes }w_{ij}(t)-r_{ij}(t)~;\\
		0 & \text{otherwise.}
	\end{cases}
	\]
	
	Correspondingly, denote the match rate of edge $(i,j)$ by \reggreedy as:
	\[y_{ij}^*(t) = \lambda_i \mu_{ij}^*(t)~.	\]
	
	Then, the first line of Eqn. \eqref{eqn:SWOR-derivative-of-A-plus-Phi_top_half} equals:
	\begin{equation*}
		\begin{multlined}
			\sum_{(i,j)\in E} w_{ij}(t) y_{ij}^*(t)
			- \alpha(t)
			\sum_{(i,j)\in E} y_{ij}^*(t) \sum_{i' \in I_j} \Big(w_{i'j}(t)-\big(w_{i'j}(t)-w_{ij}(t)\big)^+\Big) x_{i'j}\\
			- \beta(t)
			\sum_{(i,j)\in E} y_{ij}^*(t)
			\sum_{i'\in I_j} \int_{(w_{i'j}(t)-w_{ij}(t))^+}^{w_{i'j}(t)} \lambda_{i'}
			\Big( 
			p\big(\rho_{i'}(t,w)\big)
			-
			p\big(\rho_{i'}(t,w)-\rho_{i'j}(t,w)\big)
			\Big) dw\\
			=\sum_{(i,j)\in E} \big(w_{ij}(t)-r_{ij}(t)\big) y_{ij}^*(t)~,
		\end{multlined}
	\end{equation*}
	which is maximized by \reggreedy among all possible online decisions at time $t$.
	
	Hence, to prove that the derivative in Eqn. \eqref{eqn:SWOR-derivative-of-A-plus-Phi_top_half} is nonnegative, it suffices to prove it \emph{for some} online decisions at time $t$. In particular, we consider the online decisions that would have been made by \tophalf. The three terms on the first line of Eqn. \eqref{eqn:SWOR-derivative-of-A-plus-Phi_top_half} are:
	\begin{equation*}
		\begin{aligned}
			\frac{d}{dt}\E\, \big[ \alg(t) \mid w(t) \big]
			&= \sum_{(i,j)\in E} w_{ij}(t) y_{ij}(t) = Y(t)~,\\
			\frac{d}{dt} \E\, \big[X(t)\mid w(t) \big]
			&= -\sum_{(i,j)\in E} y_{ij}(t) \sum_{i' \in I_j} \Big(w_{i'j}(t)-\big(w_{i'j}(t)-w_{ij}(t)\big)^+\Big) x_{i'j}\\
			&\geq - \sum_{(i,j)\in E} y_{ij}(t) w_{ij}(t) \sum_{i' \in I_j}  x_{i'j}\\
			&\geq - \sum_{(i,j)\in E} w_{ij}(t) y_{ij}(t)\\
			&= - Y(t)~,
		\end{aligned}
	\end{equation*}
	and by Lemma \ref{lem:HuangSY-STOC-2022-restate}:
	\begin{equation*}
		\frac{d}{dt} \E\, \big[Y(t)\mid w(t)\big] \\
		\geq 2(1+\ln 2)X(t) - (3+\ln 2) Y(t)~.
	\end{equation*}
	
	Combining with $\alpha(t)\geq 0$ and $\beta(t)\geq 0$ due to Property \ref{prop: pos_0.5} of Lemma \ref{lem:alpha-beta_top_half} and merging terms, we get that:
	\begin{equation*}
		\begin{aligned}
			\eqref{eqn:SWOR-derivative-of-A-plus-Phi_top_half}
			\geq \Big(2(1+\ln 2)\beta(t)+\frac{d}{dt}\alpha(t)\Big) X(t)
			+ \Big(1-\alpha(t)-(3+\ln 2)\beta(t)+\frac{d}{dt} \beta(t) \Big) Y(t)~,
		\end{aligned}
	\end{equation*}
	where right-hand-side equals zero by Properties \ref{prop: derivative of alpha_0.5} and \ref{prop: derivative of beta_0.5} of Lemma \ref{lem:alpha-beta_top_half}. This finishes the proof.
\end{proof}

The proof of Lemma \ref{lem:HuangSY-STOC-2022-restate} requires the following lemma about function $p$.
\begin{lemma}
	\label{lem:property-of-p-for-theta=1/2}
	For any $0\leq \rho_i \leq 1$ and any $\rho_{ij}\geq 0$ such that $\sum_{j\in J_i} \rho_{ij} = \rho_i$:
	\[
	p(\rho_i) - \rho_i \geq 
	\frac12 \sum_{j\in J_i} 
	\Big(p(\rho_i)-p(\rho_i-\rho_{ij})-(2\rho_{ij}-1)^+\Big)~.
	\]
\end{lemma}
\begin{proof}
	If $\rho_i\leq \frac12$, then $p(\rho_i)=2\rho_i$ on both sides, and $(2\rho_{ij}-1)^+=0$ for all $j\in J_i$. Thus, both sides equal $\rho_i$.
	
	Otherwise, we have $\rho_i > \frac12$. The left-hand-side equals $1-\rho_i$.
	
	If $\rho_{ij} > \frac12$ for some $j$, then $\rho_i - \rho_{ij'} > \frac12$ for all other $j'\ne j$. Also, $p(\rho_i) = p(\rho_i-\rho_{ij'}) = 1$.
	Hence, the right-hand-side equals $\frac12(1-2(\rho_i-\rho_{ij})-(2\rho_{ij}-1)) = 1-\rho_i$.
	
	Otherwise, the right-hand-side is $\sum_{j\in J_i} (1-p(\rho_i-\rho_{ij}))$. By the convexity in $\rho_{ij}$, $1-p(\rho_i-\rho_{ij})$ is maximized when there are $j_1, j_2$ with $\rho_{ij_1}=\frac12$ and $\rho_{ij_2}=\rho_i - \frac12$, with maximum value $1-\rho_i$.
\end{proof}

Finally, we give the proof of Lemma \ref{lem:HuangSY-STOC-2022-restate}.

\begin{proof}[Proof of Lemma \ref{lem:HuangSY-STOC-2022-restate}]
	Recall that our goal is to prove Eqn.~\eqref{eq:restated-equation}.
	For notational simplicity, we omit $t$ from all variables in this proof except for the marginal weights.
	
	Define $\Delta_{ij} = \big(2x_{ij} - \lambda_i\big)^+$ and
	$\Delta_j = \sum_{i\in J_i} \Delta_{ij}$.
	Further for any weight-level $w\geq 0$ 
	let $\Delta_{ij}(w) = \big(2x_{ij}(w) - \lambda_i\big)^+ \leq \Delta_{ij}$.
	The definition of function $p$ implies:
	\begin{equation}
		\label{eq:relationship-beteween-p-and-Delta}
		\lambda_i p\big(\rho_i(w)\big) 
		- \lambda_i p\big(\rho_i(w)-\rho_{ij}(w)\big) - \Delta_{ij}(w) \geq 0~.
	\end{equation}
	
	By the definition of \tophalf, we have, for any weight-level $w\geq 0$:
	\begin{equation}
		\label{eq:match-rate-of-i}
		\sum_{j\in J_i} y_{ij}(w) = \lambda_i p(\rho_i(w))~,
	\end{equation}
	and
	\begin{equation}
		\label{eq:match-rate-of-(i,j)}
		\underbrace{\lambda_i p(\rho_i(w))}_{\text{match rate of}~i}
		-
		\underbrace{\lambda_i p(\rho_i(w)-\rho_{ij}(w))}_{\text{max match rate of}~(i,j')~\text{for}~j'\ne j}
		\leq 
		\underbrace{y_{ij}(w)}_{\text{match rate of}~(i,j)}
		\leq 
		\underbrace{\min\{2x_{ij}(w),\lambda_i\}}_{\text{doubled LP match rate, capped by}~\lambda_i}~.
	\end{equation}
	
	By Eqn.~\eqref{eq:match-rate-of-(i,j)}, we have:
	\begin{equation}
		\label{eq:upper-bound-of-y_j}
		y_j \leq \sum_{i\in I_j} \min\{2x_{ij},\lambda_i\} 
		= 2\sum_{i\in I_j} x_{ij} - \sum_{i\in I_j} \Delta_{ij}
		\leq 2-\Delta_j~.
	\end{equation}
	
	The third term $\frac{d}{dt} \E\, \big[Y(t)\mid w(t)\big]$ of Eqn.~\eqref{eq:restated-equation} equals:
	\begin{equation}
		\label{eq:derivative-term}
		-\sum_{(i,j)\in E} y_{ij}
		\sum_{i'\in I_j} \int_{(w_{i'j}(t)-w_{ij}(t))^+}^{w_{i'j}(t)}
		\Big( 
		\lambda_{i'} p\big(\rho_{i'}(w)\big)
		-
		\lambda_{i'} p\big(\rho_{i'}(w)-\rho_{i'j}(w)\big)
		\Big) dw~.
	\end{equation}
	
	Ignoring the leading negative sign therein for now, we bound Eqn.~\eqref{eq:derivative-term} in two parts below:
	\begin{align*}
		&\sum_{(i,j)\in E} y_{ij}
		\sum_{i'\in I_j} \int_{(w_{i'j}(t)-w_{ij}(t))^+}^{w_{i'j}(t)} 
		\Big( 
		\lambda_{i'}p\big(\rho_{i'}(w)\big)
		-
		\lambda_{i'} p\big(\rho_{i'}(w)-\rho_{i'j}(w)\big)
		- \Delta_{i'j}(w)
		\Big) dw\\
		&\leq \int_0^\infty \sum_{(i,j)\in E} y_{ij}
		\sum_{i'\in I_j} 
		\Big( 
		\lambda_{i'}p\big(\rho_{i'}(w)\big)
		-
		\lambda_{i'}p\big(\rho_{i'}(w)-\rho_{i'j}(w)\big)
		- \Delta_{i'j}(w)
		\Big) dw &\tag{Eqn.~\eqref{eq:relationship-beteween-p-and-Delta}}\\
		&\leq \int_0^\infty \sum_{j\in J} (2-\Delta_j)
		\sum_{i'\in I_j}
		\Big( 
		\lambda_{i'} p\big(\rho_{i'}(w)\big)
		-
		\lambda_{i'}p\big(\rho_{i'}(w)-\rho_{i'j}(w)\big)
		- \Delta_{i'j}(w)
		\Big) dw ~,&\tag{Eqn.~\eqref{eq:upper-bound-of-y_j}}
	\end{align*}
	and:
	\begin{align*}
		\sum_{(i,j)\in E} y_{ij}
		\sum_{i'\in I_j} \int_{(w_{i'j}(t)-w_{ij}(t))^+}^{w_{i'j}(t)} 
		\Delta_{i'j}(w) dw
		&\leq \sum_{(i,j)\in E} y_{ij}
		\sum_{i'\in I_j} w_{ij}(t)\Delta_{i'j} 
		&\tag{$\Delta_{i'j}(w)\leq \Delta_{i'j}$}\\
		&= \sum_{(i,j)\in E} y_{ij} w_{ij}(t) \Delta_j
		= \int_0^\infty \sum_{(i,j)\in E} y_{ij}(w) \Delta_j dw~.
	\end{align*}
	
	Together we have that $\frac{d}{dt} \E\, \big[Y(t)\mid w(t)\big]$ is at least:
	\begin{equation*}
		- \int_0^\infty \sum_{j\in J} \bigg((2-\Delta_j)
		\sum_{i\in I_j}
		\Big( 
		\lambda_{i} p\big(\rho_i(w)\big)
		-
		\lambda_{i}p\big(\rho_i(w)-\rho_{ij}(w)\big)
		- \Delta_{ij}(w)
		\Big) \bigg)
		- \sum_{(i,j)\in E} y_{ij}(w) \Delta_j\ dw~.
	\end{equation*}
	
	Since the other two terms $X(t)$ and $Y(t)$ of Eqn.~\eqref{eq:restated-equation}:
	\[X(t)=\int_0^\infty \sum_{(i,j)\in E} x_{ij}(w) dw~,\quad
	Y(t)=\int_0^\infty \sum_{(i,j)\in E} y_{ij}(w) dw\]
	are also integrals over weight-levels from $0$ to $\infty$, it suffices to prove for any weight-level $w\geq 0$ (rearranging terms):
	\begin{equation}
		\label{eq:diff-ineq-unweighted}
		\begin{multlined}
			\sum_{j\in J} (2-\Delta_j)
			\sum_{i\in I_j}
			\Big( 
			\lambda_{i} p\big(\rho_i(w)\big)
			-
			\lambda_{i}p\big(\rho_i(w)-\rho_{ij}(w)\big)
			- \Delta_{ij}(w)
			\Big) 
			+ \sum_{j\in J} y_{j}(w) (\Delta_j-1+\ln 2)\\
			\leq
			(2+2\ln 2) \sum_{(i,j)\in E} \big(y_{ij}(w) - x_{ij}(w)\big)~.
		\end{multlined}
	\end{equation}
	
	From now on, we also omit $w$ from all variables.
	Consider any $(z_{ij})_{(i,j)\in E}$ that satisfies 
	$\sum_{j\in J_i} z_{ij} = 2\sum_{j\in J_i} (y_{ij}-x_{ij})$ for all $i\in I$ and:
	\begin{equation}
		\label{eq:existence-of-z}
		\lambda_i p(\rho_i)
		- \lambda_i p(\rho_i-\rho_{ij})
		- \Delta_{ij}
		\leq z_{ij} \leq y_{ij}~.
	\end{equation}
	
	It exists because (i) by the first part of Eqn.~\eqref{eq:match-rate-of-(i,j)}, the lower bound of $z_{ij}$ is smaller than or equal to the upper bound; (ii) by Lemma \ref{lem:property-of-p-for-theta=1/2} and Eqn.~\eqref{eq:match-rate-of-i} the lower bound sums to at most $2\sum_{j\in J_i} (y_{ij}-x_{ij})$; and (iii) by the second part of Eqn.~\eqref{eq:match-rate-of-(i,j)} the upper bound sums to at least $2\sum_{j\in J_i} (y_{ij}-x_{ij})$.
	
	Eqn.~\eqref{eq:diff-ineq-unweighted} then reduces to:
	\[\sum_{j\in J} (2-\Delta_j) \sum_{i\in I_j} z_{ij} + \sum_{j\in J} y_j (\Delta_j -1+\ln 2)
	\leq (1+\ln 2) \sum_{(i,j)\in E} z_{ij}~.\]
	
	Rearrange terms and this becomes:
	\[
	\sum_{j\in J} (1-\ln 2-\Delta_j) \sum_{i\in I_j} (y_{ij} - z_{ij}) \geq 0~,
	\]
	which follows by Corollary \ref{cor:converse-jensen-theta} that $\Delta_j \leq 1-\ln 2$, as well as Eqn.~\eqref{eq:existence-of-z}.
\end{proof}